%% file: main.tex
\documentclass[a4paper, openany, 12pt]{article}
\input{include/preambule.tex}

\begin{document}

\title{\vspace{-0.0cm}\bf
Mathematical aspects of the decomposition of diagonal $\mathbf{U}(N)$ operators}

\author{
M. M. Fedin$^{a,c,}$\footnote{fedin.mm@iitp.ru, fedin.mm@phystech.edu},
A. A. Morozov$^{a,b,c,}$\footnote{morozov.andrey.a@iitp.ru}
}

\date{ }

\maketitle

\vspace{-7.5cm}

\begin{flushleft}
    ITEP/TH-29/25\\
    IITP/TH-22/25\\
    MIPT/TH-18/25
\end{flushleft}

\vspace{3.5cm}

\begin{center}
$^a$ {\small {\it Institute for Information Transmission Problems, Moscow 127994, Russia}}\\
$^b$ {\small {\it NRC ``Kurchatov Institute", 123182, Moscow, Russia}}\\
$^c$ {\small {\it MIPT, Dolgoprudny, 141701, Russia}}
\end{center}

\date{ }

\maketitle
    \input{parts/Annotation.tex}
    \tableofcontents{}
    \newpage

    \input{parts/Chapter0.tex} 

\input{parts/Chapter2.tex} 

\input{parts/Chapter4.1.tex} 
    \input{parts/Chapter5.tex} 
    \input{parts/Chapter6.tex} 

    \newpage
    \input{parts/Appendix.tex}

\end{document}

%% file: include/preambule.tex
\usepackage[T1]{fontenc}
\usepackage[utf8]{inputenc}

\usepackage{tikz}
\usetikzlibrary{trees, positioning, arrows.meta, decorations.pathreplacing, calc, graphs, graphs.standard}

\tikzset{
  treenode/.style = {
    circle,
    draw,
    rounded corners=4pt,
    align=center,
    minimum width=0.2cm,
    minimum height=.2cm,
    font=\sffamily,
    fill=blue!10
  },
  edge style/.style = {
    -{Stealth[length=2.5mm,width=1.8mm]},
    thick
  }
}

\usepackage{comment}

\usepackage{amsmath,amsfonts,amssymb,amsthm,mathtools}
\usepackage{icomma}
\usepackage{dsfont}
\usepackage{stmaryrd}
\usepackage{euscript}
\usepackage{mathrsfs}

\usepackage{xcolor}
\usepackage{graphicx}
\graphicspath{{images/}{images2/}}
\usepackage{wrapfig}
\usepackage{caption}
\usepackage{float}

\usepackage[left=3cm,right=1.5cm,top=2cm,bottom=2cm,bindingoffset=0cm]{geometry}

\usepackage{enumitem}

\usepackage{fancyhdr}
\makeatletter
\renewcommand\Huge{\@setfontsize\Huge{14pt}{10}}
\renewcommand\huge{\@setfontsize\huge{14pt}{10}}
\renewcommand\Large{\@setfontsize\Large{14pt}{10}}
\renewcommand\large{\@setfontsize\large{12pt}{10}}
\makeatother

\usepackage{gensymb}
\usepackage{calc}
\usepackage{indentfirst}

\usepackage[unicode, pdftex]{hyperref}

\newtheorem{lemma}{Lemma}
\newtheorem{theorem}{Theorem}
\newtheorem{prop}{Proposition}
\theoremstyle{remark}

\theoremstyle{definition}
\newtheorem{definition}{Definition}[section]

\newcommand{\beq}{\begin{equation}}
\newcommand{\eeq}{\end{equation}}
\newcommand{\beqnn}{\begin{equation*}}
\newcommand{\eeqnn}{\end{equation*}}

\newcommand{\Tr}{\operatorname{Tr}}
\newcommand{\diag}{\operatorname{diag}}

\newcommand{\sut}{\mathbf{SU}(2)}

\newcommand{\un}{\mathbf{U}(2^n)}

\newcommand{\ot}{\otimes}



%% file: parts/Annotation.tex
\begin{abstract}

    \begin{center}
    
    We prove the decomposition of arbitrary diagonal operators into tensor and matrix products of smaller matrices, focusing on the analytic structure of the resulting formulas and their inherent symmetries.  Diagrammatic representations are introduced, providing clear visualizations of the structure of these decompositions. We also discuss symmetries of the suggested decomposition. Methods and representations developed in this paper can be applied in different areas, including optimization of quantum computing algorithms, complex biological analysis, crystallography, optimization of AI models, and others.
    \end{center}

\end{abstract}

%% file: parts/Chapter0.tex
\section{Introduction \label{sec:phys}}

Tensor decompositions are widely employed across various fields of modern natural sciences for the analysis of multidimensional data. They enable the compression of large AI models \cite{Liu_2023}, facilitate the classification of quantum entanglement states \cite{doi:10.1137/24M1643451}, and support the analysis of complex multilevel biological networks \cite{doi:10.1137/23M1592547} and some highly specialized problems of crystallography \cite{aroyo2006bilbao}. Recent research also explores the use of quantum computing to develop novel algorithms for mathematical operations via tensor-based approaches \cite{Uotila_2024} or develop novel quantum architectures, such as Topological Quantum Computer (TQC) \cite{KITAEV20032},\cite{KOLGANOV2023116072},\cite{Kolganov2020} based on the knot theory \cite{BISHLER2023104729}. In certain scenarios, approximate decompositions can be performed to optimize computational complexity \cite{Schollw_ck_2011}, including through the use of tensor networks \cite{TYRTYSHNIKOV2004423}. Numerical methods are often used for practical purposes \cite{doi:10.1137/090752286}. Despite their utility, these methods encounter fundamental challenges, as finding an optimal decomposition is an NP-hard problem \cite{turchetti2023decompositionlineartensortransformations}.  For many specific applications, it is therefore advantageous to possess a general solution and to comprehend its underlying symmetries.

In this article, we will consider a recurrent approach that allows us to obtain an analytical answer for the decomposition of a diagonal matrix, as well as consider its symmetries. We will decompose the diagonal matrix from $\mathbf{U}(2^n)$ into the product of the elements of the groups $\mathbf{SU}(4)$, $\mathbf{SU}(2)$ and $\mathbf{U}(1)$.
Also, if necessary, we can use the matrices $\mathbf{U}(2)$, since they can be easily expressed as the product of the elements of the groups $\mathbf{SU}(2)$ and $\mathbf{U}(1)$. The same applies to $\mathbf{U}(4)$. In our method, the number of elements $\mathbf{SU}(4)$ is minimal for constructing a diagram in a recurrent way. For some cases, it is possible to find an algorithmically optimized solution with fewer operators $\mathbf{SU}(4)$, but there is currently no general solution. Algorithmic optimization and some of its variants are discussed in more details in the context of quantum computing in \cite{huang2024optimizedsynthesiscircuitsdiagonal}.

The main motivation for this particular decomposition basis lies in quantum computing, which uses its special case, CNOT \cite{crooks2024},\cite{nch}, instead of the $L(k)$ operator. The use of the operators $\mathbf{SU}(2)$ is motivated by the fact that, technically speaking, many quantum computer architectures allow us to perform any one-qubit operation, therefore, in addition to $CNOT$, any $\mathbf{SU}(2)$ matrices can be included in our basic operations. This approach allows you to express any operators accurately.

It is worth paying special attention to the different accuracy of one- and two-qubit operations on a real quantum computer, created in the presence of various physical limitations, the nature of which is not important to us here, however, the following is important: according to a study by \cite{expcnotsign}, typical indicators of the accuracy of executing unitary operators $\mathbf{SU}(2)$ are $99.7\%$, and $\mathbf{SU}(4)$ - $96,5\%$, this means that the error probability differs by about an order of magnitude, which requires us to minimize the number of operators from $\mathbf{SU}(4)$ in our decomposition. From a mathematical point of view, the accuracy of execution is conveniently interpreted as the probability of non-destruction of the quantum state (entanglement of qubits) after applying the operator.

Experts who are more familiar with quantum computing may also notice that the decomposition into universal bases, such as the set of unitary operators $\{H,T, CNOT\}$, where $H =\frac{1}{\sqrt{2}}\begin{bmatrix}
    1&-1\\1&1
\end{bmatrix}$, a $T =\begin{bmatrix}
    1 &0\\0&e^{i\frac{\pi}{4}}
\end{bmatrix}$, assumes the portability of quantum algorithms between different hardware platforms.    the group of unitary matrices is continuous, and the group generated by a finite number of operators is no more than countable. 

This is especially important in the context of the existence of various architectures of quantum processors. In other words, if we know how to simulate the operators $\{H,T, CNOT\}$ on a given quantum processor, then we know how to execute any quantum algorithm on this quantum processor with a given accuracy, according to the Solovay-Kitaev theorem \cite{kitaev1997}, \cite{kitaev2002}. In our work, decompositions with a given accuracy are not considered.

Moreover, it is clear that if a decomposition exists for a matrix from $\mathbf{U}(n)$ for any $n \in \mathbb{N}$, then for any given $N$, one can find an $n$ such that the decomposition for a matrix in $\mathbf{U}(N)$ is obtained as a trivial reduction of the former. In Appendix \ref{sec:suNtosu2n}, we carefully prove
 that the ability to decompose matrices of size $2^n$ implies the ability to decompose matrices of any size $N < 2^n$.

\section{Results and definitions}

In this paper we prove the following theorem:

\begin{theorem}[About reccurent decomposition\label{th:reccurent}]
It is always possible to decompose diagonal $D_n$ matrix using reccurent formula:
\begin{multline}
        D_{n}(\alpha_1, \alpha_2, \dots, \alpha_{2^n})=  \Bigl(D_{n-1}(\bar\alpha_1, \bar\alpha_2, \dots, \bar\alpha_{2^{n-1}}) \ot I\Bigr) \cdot U_{tail} =\\= \Bigl(D_{n-1}(\bar\alpha_1, \bar\alpha_2, \dots, \bar\alpha_{2^{n-1}}) \ot I\Bigr) \cdot \prod_{i=1}^{2^{n-1}} \biggl( \Bigl(I_{2^{n-1}} \ot D_1(\beta_i, -\beta_i)\Bigr) \cdot L(A_n(i)) \biggr)
    \label{eq:certralreccurentformula}
\end{multline}

With linear bijection between parameters:
\begin{align}
    \mathcal{L}: (\bar\alpha_1, \bar\alpha_2, \dots, \bar\alpha_{2^{n-1}},\beta_1,\beta_2,\dots,\beta_{2^{n-1}}) \rightarrow (\alpha_1, \alpha_2, \dots, \alpha_{2^n}) \\
    \mathcal{L}^{-1}: (\bar\alpha_1, \bar\alpha_2, \dots, \bar\alpha_{2^{n-1}},\beta_1,\beta_2,\dots,\beta_{2^{n-1}}) \leftarrow (\alpha_1, \alpha_2, \dots, \alpha_{2^n})
\end{align}
We will demonstrate which types of $A_n$ are suitable, the exact definition of which will be given later in this section.
Also all of the quantities: $L(k),\ A_n(i),\ D_n$ will be defined below in this section.
\end{theorem}
Here, $I$ denotes the identity matrix in $\mathbf{SU}(2)$, and $I_N$ is the identity matrix in $\mathbf{SU}(N)$.

As we can see below, this is the best recursive theoretical decomposition with the same number of $L(k)$ operators as the scheme described in \cite{1629135} and \cite{crooks2024}. However, our scheme was obtained in a different way and is a generalization of the methods described in these papers.

$\mathcal{L}$ and $\mathcal{L}^{-1}$ will be constructed in Section \ref{sec:fullproof}.

\begin{definition}[Matrix $D_n$]
\label{def:D_n}
$D_n(\alpha_1, \alpha_2, \dots, \alpha_{2^n}) \in \mathbf{U}(2^n)$ defined by the formula:
\begin{equation}
D_{n}(\alpha_1, \alpha_2, \dots, \alpha_{2^n}) = \diag(e^{i\alpha_1},e^{i\alpha_2}, \dots, e^{i\alpha_{2^n}})
\end{equation}
Sometimes we will write $D_1(\varphi)$ using just one argument. This notation will be equivalent to:
\begin{equation}
D_1(\varphi) = D_1(\varphi,-\varphi)
\end{equation}
Also note, that $D_0(\alpha_1) = e^{i\alpha_1}$.
\end{definition}

\begin{definition}[Matrix $X$]
\label{def:X}
 $X\in\mathbf{U}(2)$ is a matrix satisfying the following properties:
\begin{equation}
X \in\mathbf{U}(2),\quad X^2 = I,\quad \Tr(IX) = 0,\quad \Tr(ZX)=0,\text{ where } Z = \begin{bmatrix}
1&0\\
0&-1
\end{bmatrix}
\end{equation}
\end{definition}

It should be noted that for numerical calculations in specific cases, any matrix $X$ permitted by the definition can be used, for example $X = \begin{bmatrix} 0&1\\1&0\end{bmatrix}$. This does not affect either the decomposition parameters or, in particular, its non-degeneracy. This follows from the fact that we will only need the sign-flipping property of a diagonal matrix via the matrix $X$, see formulas \eqref{eq:D_prop}.

\begin{definition}[Control Matrix $L(k)$]
\label{def:Lk}
The control matrix $L(k)$ is defined as the matrix obtained from the sum of two matrices: $I^{\ot (k-1)} \ot \begin{bmatrix}
1&0\\
0&0
\end{bmatrix} \ot I^{\ot (n-k)}$ and $I^{\ot (k-1)} \ot \begin{bmatrix}
0&0\\
0&1
\end{bmatrix} \ot I^{\ot (n-k-1)}\ot X$:
\begin{equation}
L(k) = I^{\ot (k-1)} \ot \begin{bmatrix}
1&0\\
0&0
\end{bmatrix} \ot I^{\ot (n-k)} +
I^{\ot (k-1)} \ot \begin{bmatrix}
0&0\\
0&1
\end{bmatrix} \ot I^{\ot (n-k-1)}\ot X
\label{eq:Lclose}
\end{equation}
\end{definition}

It is convenient here and further to use the projector notation for compactness:
\begin{align}
\pi_0 = \begin{bmatrix}
1&0\\
0&0
\end{bmatrix},&\quad \pi_1 = \begin{bmatrix}
0&0\\
0&1
\end{bmatrix}\\
\pi_0^2 = \pi_0,&\quad \pi_1^2 = \pi_1;\\
\pi_0I = \pi_0,&\quad \pi_1I = \pi_1
\end{align}

In this form, $L(k)$ can be rewritten as:
\begin{equation}
L(k) = I^{\ot (k-1)} \ot \pi_0\ot I^{\ot (n-k)} +
I^{\ot (k-1)} \ot\pi_1 \ot I^{\ot (n-k-1)}\ot X
\label{eq:l_k_throw_pi}
\end{equation}

In this representation, it becomes apparent that with appropriate renumbering of the tensor spaces, this is equivalent to:
\begin{equation}
\widetilde{L}(k) = I^{\ot (n-2)} \ot (\pi_0 \ot I + \pi_1 \ot X)
\end{equation}

Which shows that the nontrivial part of the operator $L(k)$ lies in $SU(4)$.

\begin{prop}{$\forall k,m \in {1,\dots,n-1}: [L(k),L(m)] = 0$\label{prop:L_commut}}
\end{prop}
\begin{proof}

\begin{enumerate}
\item $k=m$ — trivial, since every operator commutes with itself.
\item $k < m$.
   Substituting formula \eqref{eq:l_k_throw_pi} into the commutator we obtain

\begin{multline}
    L(k) = I^{\ot (k-1)} \ot \pi_0\ot I^{\ot (m-k)}\ot I^{\ot (n-m-1)}\ot I + \\
       + I^{\ot (k-1)} \ot \pi_1\ot I^{\ot (m-k)}\ot I^{\ot (n-m-1)}\ot X 
\end{multline}
\begin{multline}
    L(m) = I^{\ot (k-1)} \ot I^{\ot (m-k)}\ot \pi_0\ot I^{\ot (n-m-1)}\ot I + \\
       + I^{\ot (k-1)} \ot I^{\ot (m-k)}\ot \pi_1\ot I^{\ot (n-m-1)}\ot X 
\end{multline}

\begin{multline}
L(k)L(m) = 
I^{\ot(k-1)} \ot \boldsymbol{\pi_0} \ot I^{\ot(m-k-1)} \ot \boldsymbol{\pi_0} \ot I^{\ot(n-m-1)} \ot \boldsymbol{I} \quad + \\
+ I^{\ot(k-1)} \ot \boldsymbol{\pi_0} \ot I^{\ot(m-k-1)} \ot \boldsymbol{\pi_1} \ot I^{\ot(n-m-1)} \ot \boldsymbol{X} \quad + \\
+ I^{\ot(k-1)} \ot \boldsymbol{\pi_1} \ot I^{\ot(m-k-1)} \ot \boldsymbol{\pi_0} \ot I^{\ot(n-m-1)} \ot \boldsymbol{X} \quad + \\
+ I^{\ot(k-1)} \ot \boldsymbol{\pi_1} \ot I^{\ot(m-k-1)} \ot \boldsymbol{\pi_1} \ot I^{\ot(n-m-1)} \ot \boldsymbol{I} \quad\\
\end{multline}

\begin{multline}
L(m)L(k) = 
I^{\ot(k-1)} \ot \boldsymbol{\pi_0} \ot I^{\ot(m-k-1)} \ot \boldsymbol{\pi_0} \ot I^{\ot(n-m-1)} \ot \boldsymbol{I} \quad + \\
+ I^{\ot(k-1)} \ot \boldsymbol{\pi_1} \ot I^{\ot(m-k-1)} \ot \boldsymbol{\pi_0} \ot I^{\ot(n-m-1)} \ot \boldsymbol{X} \quad + \\
+ I^{\ot(k-1)} \ot \boldsymbol{\pi_0} \ot I^{\ot(m-k-1)} \ot \boldsymbol{\pi_1} \ot I^{\ot(n-m-1)} \ot \boldsymbol{X} \quad + \\
+ I^{\ot(k-1)} \ot \boldsymbol{\pi_1} \ot I^{\ot(m-k-1)} \ot \boldsymbol{\pi_1} \ot I^{\ot(n-m-1)} \ot \boldsymbol{I} \quad\\
\end{multline}

Therefore:

\begin{equation}
[L(k),L(m)] = L(k)L(m) - L(m)L(k) = 0
\end{equation}

\item $k>m$. $[L(k),L(m)] = 0 \Leftrightarrow [L(m),L(k)] = 0$ and this relation completes the proof.
\end{enumerate}

\end{proof}

\begin{definition}[Sequence $A_n(i)$]
This is a sequence of $2^{n-1}$ natural numbers: $A_n = \{a_i \in \{1,2,\dots,n-1\}: \forall i \in \{1,2,\dots,2^{n-1}\}\}$
\end{definition}

We will prove that, the form of our answer is in accordance with the formula~\eqref{eq:certralreccurentformula} without searching for specific parameters in the Section~\ref{sec:Chapter5}; the derivation of the specific linear mapping will be given in Section~\ref{sec:r_formula}; symmetries of the solution and how the parametric mapping changes will be discussed in Section~\ref{sec:symmetries}.

Note that we use the notation $\vec{x}$ in both cases: when $x$ is a vector and when $x$ is a column. There is no confusion, because the context always unambiguously defines a column or a vector, or a specific interpretation does not change the meaning of the expressions.

%% file: parts/Chapter5.tex
\section{Validity of the Decomposition Form}
\label{sec:Chapter5} \index{Chapter5}

In this section we discuss form of $U_{tail}$ operator from the formula \eqref{eq:certralreccurentformula}. We will describe some properties of this operator and its decomposition. We will describe the mapping between the parameters of this operator and the parameters in the decomposition and prepare the necessary basis for a deeper study of this mapping in the Section \ref{sec:r_formula}.

\subsection{Necessary conditions on the $U_{tail}$ \label{sec:CNOTSCHEMEOUR}}

By examining~\eqref{eq:certralreccurentformula}, we can deduce the structure that $U_{tail}$ must have.

\begin{prop}[$U_{tail}$ is diagonal]
\label{prop:U_tail_is_diag}
\end{prop}
\begin{proof}
In the formula $D_n = (D_{n-1} \ot I)\cdot U_{tail}$, both $D_n$ and $(D_{n-1} \ot I)$ are diagonal matrices. We can then write the full form of $D_n$ by substituting definition~\ref{def:D_n}:
\begin{align}
    \diag(e^{i\alpha_1}, e^{i\alpha_2}, \dots, e^{i\alpha_{2^n}}) = \diag(e^{i\bar\alpha_1}, e^{i\bar\alpha_2}, \dots, e^{i\bar\alpha_{2^{n-1}}})\ot\diag(1,1) \cdot U_{tail} \\
    \diag(e^{i\alpha_1}, e^{i\alpha_2}, \dots, e^{i\alpha_{2^n}}) = \diag(e^{i\bar\alpha_1},e^{i\bar\alpha_1}, e^{i\bar\alpha_2},e^{i\bar\alpha_2}, \dots, e^{i\bar\alpha_{2^{n-1}}}, e^{i\bar\alpha_{2^{n-1}}}) \cdot U_{tail} \\
    \diag(e^{i\alpha_1}, e^{i\alpha_2}, \dots, e^{i\alpha_{2^n}})\diag(e^{-i\bar\alpha_1},e^{-i\bar\alpha_1}, e^{-i\bar\alpha_2},e^{-i\bar\alpha_2}, \dots, e^{-i\bar\alpha_{2^{n-1}}}, e^{-i\bar\alpha_{2^{n-1}}}) = U_{tail} \\
    \diag(e^{i\alpha_1-i\bar\alpha_1}, e^{i\alpha_2-i\bar\alpha_1}, \dots, e^{i\alpha_{2^n}-i\bar\alpha_{2^{n-1}}}) = U_{tail}
\end{align}
This final expression completes the proof.
\end{proof}

We can then represent $U_{tail}$ as:

\begin{equation}
    U_{tail} = \diag(e^{i\gamma_1}, e^{-i\gamma_1+i\gamma'_1}, \dots, e^{i\gamma_{2^{n-1}}}, e^{-i\gamma_{2^{n-1}}+i\gamma'_{2^{n-1}}})
    \label{eq:U_tail_general}
\end{equation}

This is the most general form of a diagonal matrix from $U(2^{n-1})$. If $\forall i \in \{1,2,3,\dots,2^{n-1}\}: \gamma'_i = 0$, then $U_{tail}$ is a block-diagonal matrix with $D_1(\gamma_i)\in SU(2)$ blocks on its diagonal. This observation will be useful later.

Moreover, since both matrices $U_{tail}$ and $D_{n-1}$ are diagonal, we can conveniently treat diagonal matrices as vectors. To simplify the expressions, we introduce the following notation:

\begin{equation}
    \diag(e^{i\kappa_1},\dots,e^{i\kappa_m}) \equiv (\kappa_1,\dots,\kappa_m)
    \label{notation}
\end{equation}

With this notation, multiplication becomes addition and vice versa:

\begin{multline}
    \diag(e^{i\kappa_1},\dots,e^{i\kappa_m})\cdot \diag(e^{i\mu_1},\dots,e^{i\mu_m})
    \equiv (\kappa_1,\dots,\kappa_m)\cdot(\mu_1,\dots,\mu_m) =\\
    = (\kappa_1+\mu_1,\dots,\kappa_m+\mu_m)
\end{multline}

From now on, any row in parentheses denotes a diagonal matrix constructed according to \eqref{notation}, whereas a row in square brackets denotes an ordinary row vector.

We introduce the following notations: $\vec{\bar{\alpha}}$ for $D_{n-1}$, $\vec{\gamma}$ and $\vec{\gamma'}$ for the parameters of $U_{tail}$, and $\vec{\alpha}$ for the final parameters of $D_n$. Using our notation (\ref{notation}):

\begin{multline}
 \Big((\bar{\alpha}_1,\dots,\bar{\alpha}_{2^{n-1}}) \otimes (1,1)\Big) \cdot \Big((\gamma_1,\dots,\gamma_{2^{n-1}}) \otimes (1,-1)\Big)\cdot\Big((\gamma'_1,\dots,\gamma'_{2^{n-1}}) \otimes (0,1)\Big) = \\
 = (\alpha_1, \dots, \alpha_{2^n})
\end{multline}

$U_{tail}$ can only have $2^{n-1}$ free parameters. The remaining ones must depend on them. If we set $\forall i\in\{1,2,3,\dots,2^{n-1}\}: \gamma'_i=0$, then the number of remaining parameters will be exactly $2^{n-1}$, and the matrix itself becomes special unitary.
In that case, the decomposition succeeds, and we obtain a system of equations:

\begin{equation}
[\bar{\alpha}_1,\dots,\bar{\alpha}_{2^{n-1}}]\otimes[1,1] + [\gamma_1,\dots,\gamma_{2^{n-1}}]\otimes[1,-1] = [\alpha_1,\dots,\alpha_{2^n}] \end{equation}
The solution of this system is:
\begin{equation}     \bar{\alpha}_i = \frac{\alpha_{2i-1}+\alpha_{2i}}{2} \end{equation}
\begin{equation}     \gamma_i = \frac{\alpha_{2i-1}-\alpha_{2i}}{2} \end{equation}

This solution is non-degenerate with respect to $\vec{\bar\alpha}$ and $\vec\gamma$. It is worth noting that the equality remains valid even when the parameter values go beyond the interval $[-\pi, \pi]$, despite the fact that these parameters are exponents of complex exponentials.

If necessary, one can add a term $g_i \cdot 2\pi$ where $g_i\in\mathbb{Z}$ to each resulting parameter to bring it back within the interval $[-\pi, \pi]$.

Also we can make the first step of constructing $\mathcal{L}$:

\begin{equation}
    \mathcal{L}: [\bar{\alpha}_1,\dots,\bar{\alpha}_{2^{n-1}}]\otimes[1,1] + [\gamma_1(\vec\beta),\dots,\gamma_{2^{n-1}}(\vec\beta)]\otimes[1,-1] = [\alpha_1,\dots,\alpha_{2^n}]
\end{equation}

To build the $\mathcal(L)$ mapping completely, we need to build the $\gamma(\vec\beta)$ mapping. Now, let us write down the full formula for $U_{tail}$ using \eqref{eq:certralreccurentformula}:

\begin{equation}
    U_{tail} = \prod_{i=1}^{2^{n-1}} \biggl( \Bigl(I_{2^{n-1}} \ot D_1(\beta_i)\Bigr) \cdot L(A_n(i)) \biggr)
    \label{eq:U_tail_def}
\end{equation}

Using definitions \ref{def:D_n} and \ref{def:X}, we can obtain:

\begin{equation}
    D_1(\varphi_1)D_1(\varphi_2) = D_1(\varphi_1+\varphi_2) \quad XD_1(\varphi)X = D_1(-\varphi)
    \label{eq:D_prop}
\end{equation}

In the next Section \ref{sec:NCOTAN} we consider the necessary conditions that the sequence $A_n$ must satisfy. Of all possible $A_n$ satisfying these conditions, we can choose the one for which the proof is simple, without loss of generality.  All other $A_n$,  allowed in such a recurrent decomposition, can be constructed using symmetries which we will discuss in Section \ref{sec:symmetries}.

\subsection{Necessary conditions on the $A_n$\label{sec:NCOTAN}}

By substituting the values of $L(A_n(i))$ from Definition \ref{def:Lk}, we obtain:

\begin{equation}
    U_{tail} = \prod_{i=1}^{2^{n-1}} I^{\ot (A_n(i)-1)} \ot \biggl(\pi_0 \ot I^{\ot (n-A_n(i)-1)} \ot D_1(\beta_i) +
         \pi_1 \ot I^{\ot (n-A_n(i)-1)}\ot D_1(\beta_i)X \biggr)
    \label{eq:U_tail}
\end{equation}

It is easy to see, that, the matrices $\pi_1 \ot I^{\ot (n-A_n(i)-1)}\ot D_1(\beta_i)X$ are not diagonal for any $i$.

This imposes three necessary but not sufficient conditions on the sequence $A_n$ in order for $U_{tail}$ to be diagonal.

First, by definition, the sequence $A_n$ contains an exponential number of elements: $2^{n-1}$, as we can see in the formula \eqref{eq:U_tail}. Second, their values lie among the integers $\{1, 2, \dots, n-1\}$. Third, each value appears an even number of times in the sequence $A_n$.

The first two conditions follow directly from the definition of $A_n$, while the last one will be proven in the Prop. \ref{prop:pairityX}.

\begin{prop}[On the parity of $X$ operators\label{prop:pairityX}]
If $\exists m$ such that $\sum_{i=1}^{2^{n-1}}\delta_{m,A_n(i)} \mod 2 = 1$, then $U_{tail}$ is non-diagonal.
\end{prop}

\begin{proof}
Let $\sum_{i=1}^{2^{n-1}}\delta_{m,A_n(i)} = l$.
By reordering the tensor spaces, we can always reduce the problem to the case $m = n - 1$. Then, in the full product we multiply matrices of the form:

\begin{equation}
    I^{\ot(n-2)}\ot\begin{bmatrix}
        D_1(\beta_i)&0\\0&D_1(\beta_i)X
    \end{bmatrix}
\end{equation}

As well as matrices of the form:

\begin{equation}
    I^{\ot(n-1-k)}\ot\left(\pi_0\ot I^{\ot k}\begin{bmatrix}
        D_1(\beta_i)&0\\0&D_1(\beta_i)
    \end{bmatrix} + \pi_1\ot I^{\ot k}\begin{bmatrix}
        D_1(\beta_i)X&0\\0&D_1(\beta_i)X
    \end{bmatrix}\right)
\end{equation}

Since in the second case the number of $X$ operators in each cell on the diagonal of the final block-diagonal matrix is the same, they can change the number of $X$ operators in each respective cell by the same amount. Let the total shift in the number of $X$ operators be $b$. Then the final shift on each diagonal block will be:

\begin{equation}
    \begin{bmatrix}
        X^{b}&0\\0&X^{b+l}
    \end{bmatrix} = \begin{bmatrix}
        X^{b}&0\\0&X^{b+(l\mod 2)}
    \end{bmatrix} = \begin{bmatrix}
        X^{b}&0\\0&X^{b+1}
    \end{bmatrix}
\end{equation}

Here we omit the $D_1$ operators for brevity and clarity, as they do not affect diagonalness and would clutter the notation.
In such a configuration, either $b$ is even (making the upper block diagonal and the lower one non-diagonal), or $b+1$ is even (making the lower block diagonal and the upper one non-diagonal).
This observation completes the proof.
\end{proof}

Note that all $m$ are not related to each other, since they belong to different tensor spaces, which means that the Prop.~\ref{prop:pairityX} remains valid for any number of $m$ suitable for its conditions.

\subsection{Linearity of the $\mathcal{L}$ mapping\label{sec:linearityL}}

Now we are ready to find $\vec\gamma(\vec\beta)$ for achieve constructing of $\mathcal{L}$ map. We also prove linearity of map $\vec\gamma(\vec\beta)$ and linearity $\mathcal{L}$, as a consequence of linearity $\vec\gamma(\vec\beta)$. First of all, let's choose appropriate sequence $A_n$ for our goals.
The simplest sequence $A_n$ can be generated automatically by this condition. For this, we introduce the concept of a Perfect Binary Tree (PBT):

 \begin{definition}(Perfect Binary Tree (PBT)\label{def:PBT})
The set $\mathcal{P}_h$ of perfect binary trees of height $h \in \mathbb{N}$ is defined recursively. A tree is a pair $T=(V,E)$ of vertices and edges.
For $h=1$, $\mathcal{P}_1$ contains the single tree $T_1 = (\{v_1\}, \emptyset)$.
For $h>1$, a tree $T=(V,E) \in \mathcal{P}_h$ is formed from a new root vertex $v$ and two trees $T_L=(V_L, E_L)$ and $T_R=(V_R, E_R)$ from $\mathcal{P}_{h-1}$ with $V_L \cap V_R = \emptyset$, such that
\begin{align}
    V &= V_L \cup \{v\} \cup V_R, \\
    E &= E_L \cup \{ (v, v_L), (v, v_R) \} \cup E_R ,
\end{align}
where $v_L$ and $v_R$ are the roots of $T_L$ and $T_R$ respectively.
\end{definition}

The graphical representation of this os presented on Fig. \ref{fig:PBT_steps} and more details about binary trees can be found in \cite{harder_ece250_pbt} and \cite{Knuth1997TAOCP1}.

\begin{figure}
\centering
\scalebox{0.75}{%
\begin{tabular}{cccc}
\begin{tikzpicture}[arrowstyle/.style={-{Stealth[length=2mm,width=2mm]}, thick, draw=black!70}]
  \node[treenode] (root) {$1$};
  \coordinate (tRoot) at ($(root.south)+(0,-10mm)$);
  \draw[arrowstyle] (root.south) -- (tRoot);
  \node[anchor=north] at (tRoot) {$\{1\}$};
\end{tikzpicture}

&
\begin{tikzpicture}[level distance=10mm,
  level 1/.style={sibling distance=10mm},
  arrowstyle/.style={-{Stealth[length=2mm,width=2mm]}, thick, draw=black!70}
  ]
  \node[treenode] (root) {$1$}
    child {node[treenode] (L) {$2$}}
    child {node[treenode] (R) {$2$}};
  \coordinate (tL)    at ($(L.south)+(0,-7mm)$);
  \coordinate (tRoot) at ($(root.south)+(0,-17mm)$);
  \coordinate (tR)    at ($(R.south)+(0,-7mm)$);
  \draw[arrowstyle] (L.south)    -- (tL);
  \draw[arrowstyle] (root.south) -- (tRoot);
  \draw[arrowstyle] (R.south)    -- (tR);
  \node[anchor=north] at (tL)    {$\{2,$};
  \node[anchor=north] at (tRoot) {$1$};
  \node[anchor=north] at (tR)    {$, 2\}$};
\end{tikzpicture}

&
\begin{tikzpicture}[level distance=10mm,
  level 1/.style={sibling distance=20mm},
  level 2/.style={sibling distance=10mm},
  arrowstyle/.style={-{Stealth[length=2mm,width=2mm]}, thick, draw=black!70}
  ]

  \node[treenode] (root) {$1$}
    child {node[treenode] (L) {$2$}
      child {node[treenode] (LL) {$3$}}
      child {node[treenode] (LR) {$3$}}
    }
    child {node[treenode] (R) {$2$}
      child {node[treenode] (RL) {$3$}}
      child {node[treenode] (RR) {$3$}}
    };
  \coordinate (tLL) at ($(LL.south)+(0,-7mm)$);
  \coordinate (tLR) at ($(LR.south)+(0,-7mm)$);
  \coordinate (tL)  at ($(L.south)+(0,-17mm)$);
  \coordinate (tR)  at ($(R.south)+(0,-17mm)$);
  \coordinate (tRL) at ($(RL.south)+(0,-7mm)$);
  \coordinate (tRR) at ($(RR.south)+(0,-7mm)$);
  \coordinate (tRoot) at ($(root.south)+(0,-27mm)$);
  \draw[arrowstyle] (LL.south) -- (tLL);
  \draw[arrowstyle] (LR.south) -- (tLR);
  \draw[arrowstyle] (L.south)  -- (tL);
  \draw[arrowstyle] (R.south)  -- (tR);
  \draw[arrowstyle] (RL.south) -- (tRL);
  \draw[arrowstyle] (RR.south) -- (tRR);
  \draw[arrowstyle] (root.south) -- (tRoot);
  \node[anchor=north] at (tLL) {$\{3, $};
  \node[anchor=north] at (tLR) {$3, $};
  \node[anchor=north] at (tL)  {$2, $};
  \node[anchor=north] at (tR)  {$2, $};
  \node[anchor=north] at (tRL) {$3,$};
  \node[anchor=north] at (tRR) {$3\}$};
  \node[anchor=north] at (tRoot) {$ 1,$};
\end{tikzpicture}

&
\begin{tikzpicture}[level distance=10mm,
  level 1/.style={sibling distance=40mm},
  level 2/.style={sibling distance=20mm},
  level 3/.style={sibling distance=10mm},
  arrowstyle/.style={-{Stealth[length=2mm,width=2mm]}, thick, draw=black!70}
  ]
  \node[treenode] (root) {$1$}
    child {node[treenode] (L) {$2$}
      child {node[treenode] (L1) {$3$}
        child {node[treenode] (L11) {$4$}}
        child {node[treenode] (L12) {$4$}}
      }
      child {node[treenode] (L2) {$3$}
        child {node[treenode] (L21) {$4$}}
        child {node[treenode] (L22) {$4$}}
      }
    }
    child {node[treenode] (R) {$2$}
      child {node[treenode] (R1) {$3$}
        child {node[treenode] (R11) {$4$}}
        child {node[treenode] (R12) {$4$}}
      }
      child {node[treenode] (R2) {$3$}
        child {node[treenode] (R21) {$4$}}
        child {node[treenode] (R22) {$4$}}
      }
    };
  \coordinate (tL11) at ($(L11.south)+(0,-7mm)$);
  \coordinate (tL12) at ($(L12.south)+(0,-7mm)$);
  \coordinate (tL21) at ($(L21.south)+(0,-7mm)$);
  \coordinate (tL22) at ($(L22.south)+(0,-7mm)$);
  \coordinate (tR11) at ($(R11.south)+(0,-7mm)$);
  \coordinate (tR12) at ($(R12.south)+(0,-7mm)$);
  \coordinate (tR21) at ($(R21.south)+(0,-7mm)$);
  \coordinate (tR22) at ($(R22.south)+(0,-7mm)$);

  \coordinate (tL1) at ($(L1.south)+(0,-17mm)$);
  \coordinate (tL2) at ($(L2.south)+(0,-17mm)$);
  \coordinate (tR1) at ($(R1.south)+(0,-17mm)$);
  \coordinate (tR2) at ($(R2.south)+(0,-17mm)$);

  \coordinate (tL)  at ($(L.south)+(0,-27mm)$);
  \coordinate (tR)  at ($(R.south)+(0,-27mm)$);
  \coordinate (tRoot) at ($(root.south)+(0,-37mm)$);
  \foreach \a/\b in {L11/tL11,L12/tL12,L21/tL21,L22/tL22,
                     R11/tR11,R12/tR12,R21/tR21,R22/tR22,
                     L1/tL1,L2/tL2,R1/tR1,R2/tR2,
                     L/tL,R/tR,root/tRoot}
    \draw[arrowstyle] (\a.south) -- (\b);
  \node[anchor=north] at (tL11) {$\{4,$};
  \node[anchor=north] at (tL12) {$4,$};
  \node[anchor=north] at (tL21) {$4,$};
  \node[anchor=north] at (tL22) {$4,$};

  \node[anchor=north] at (tR11) {$4,$};
  \node[anchor=north] at (tR12) {$4,$};
  \node[anchor=north] at (tR21) {$4,$};
  \node[anchor=north] at (tR22) {$4\}$};

  \node[anchor=north] at (tL1) {$3,$};
  \node[anchor=north] at (tL2) {$3,$};
  \node[anchor=north] at (tR1) {$3,$};
  \node[anchor=north] at (tR2) {$3,$};

  \node[anchor=north] at (tL)  {$2,$};
  \node[anchor=north] at (tR)  {$2,$};

  \node[anchor=north] at (tRoot) {$ 1,$};

\end{tikzpicture}

\\[2em]

\multicolumn{4}{c}{
\begin{tikzpicture}[level distance=10mm,
  level 1/.style={sibling distance=80mm},
  level 2/.style={sibling distance=40mm},
  level 3/.style={sibling distance=20mm},
  level 4/.style={sibling distance=10mm},
  arrowstyle/.style={-{Stealth[length=2mm,width=2mm]}, thick, draw=black!70}
  ]

  \node[treenode] (root) {$1$}
    child {node[treenode] (L) {$2$}
      child {node[treenode] (L1) {$3$}
        child {node[treenode] (L11) {$4$}
          child {node[treenode] (L111) {$5$}}
          child {node[treenode] (L112) {$5$}}
        }
        child {node[treenode] (L12) {$4$}
          child {node[treenode] (L121) {$5$}}
          child {node[treenode] (L122) {$5$}}
        }
      }
      child {node[treenode] (L2) {$3$}
        child {node[treenode] (L21) {$4$}
          child {node[treenode] (L211) {$5$}}
          child {node[treenode] (L212) {$5$}}
        }
        child {node[treenode] (L22) {$4$}
          child {node[treenode] (L221) {$5$}}
          child {node[treenode] (L222) {$5$}}
        }
      }
    }
    child {node[treenode] (R) {$2$}
      child {node[treenode] (R1) {$3$}
        child {node[treenode] (R11) {$4$}
          child {node[treenode] (R111) {$5$}}
          child {node[treenode] (R112) {$5$}}
        }
        child {node[treenode] (R12) {$4$}
          child {node[treenode] (R121) {$5$}}
          child {node[treenode] (R122) {$5$}}
        }
      }
      child {node[treenode] (R2) {$3$}
        child {node[treenode] (R21) {$4$}
          child {node[treenode] (R211) {$5$}}
          child {node[treenode] (R212) {$5$}}
        }
        child {node[treenode] (R22) {$4$}
          child {node[treenode] (R221) {$5$}}
          child {node[treenode] (R222) {$5$}}
        }
      }
    };
  \foreach \n/\c in {L111/tL111,L112/tL112,L121/tL121,L122/tL122,
                     L211/tL211,L212/tL212,L221/tL221,L222/tL222,
                     R111/tR111,R112/tR112,R121/tR121,R122/tR122,
                     R211/tR211,R212/tR212,R221/tR221,R222/tR222}
    \coordinate (\c) at ($(\n.south)+(0,-7mm)$);

  \foreach \n/\c in {L11/tL11,L12/tL12,L21/tL21,L22/tL22,
                     R11/tR11,R12/tR12,R21/tR21,R22/tR22}
    \coordinate (\c) at ($(\n.south)+(0,-17mm)$);

  \foreach \n/\c in {L1/tL1,L2/tL2,R1/tR1,R2/tR2}
    \coordinate (\c) at ($(\n.south)+(0,-27mm)$);

  \foreach \n/\c in {L/tL,R/tR}
    \coordinate (\c) at ($(\n.south)+(0,-37mm)$);
  \coordinate (tRoot) at ($(root.south)+(0,-47mm)$);

  \foreach \a/\b in {L111/tL111,L112/tL112,L121/tL121,L122/tL122,
                     L211/tL211,L212/tL212,L221/tL221,L222/tL222,
                     R111/tR111,R112/tR112,R121/tR121,R122/tR122,
                     R211/tR211,R212/tR212,R221/tR221,R222/tR222,
                     L11/tL11,L12/tL12,L21/tL21,L22/tL22,
                     R11/tR11,R12/tR12,R21/tR21,R22/tR22,
                     L1/tL1,L2/tL2,R1/tR1,R2/tR2,
                     L/tL,R/tR,root/tRoot}
    \draw[arrowstyle] (\a.south) -- (\b);

  \node[anchor=north] at (tL111) {$\{5,$};
  \node[anchor=north] at (tL112) {$5,$};
  \node[anchor=north] at (tL121) {$5,$};
  \node[anchor=north] at (tL122) {$5,$};
  \node[anchor=north] at (tL211) {$5,$};
  \node[anchor=north] at (tL212) {$5,$};
  \node[anchor=north] at (tL221) {$5,$};
  \node[anchor=north] at (tL222) {$5,$};

  \node[anchor=north] at (tR111) {$5,$};
  \node[anchor=north] at (tR112) {$5,$};
  \node[anchor=north] at (tR121) {$5,$};
  \node[anchor=north] at (tR122) {$5,$};
  \node[anchor=north] at (tR211) {$5,$};
  \node[anchor=north] at (tR212) {$5,$};
  \node[anchor=north] at (tR221) {$5,$};
  \node[anchor=north] at (tR222) {$5\}$};

  \node[anchor=north] at (tL11) {$4,$};
  \node[anchor=north] at (tL12) {$4,$};
  \node[anchor=north] at (tL21) {$4,$};
  \node[anchor=north] at (tL22) {$4,$};
  \node[anchor=north] at (tR11) {$4,$};
  \node[anchor=north] at (tR12) {$4,$};
  \node[anchor=north] at (tR21) {$4,$};
  \node[anchor=north] at (tR22) {$4,$};

  \node[anchor=north] at (tL1) {$3,$};
  \node[anchor=north] at (tL2) {$3,$};
  \node[anchor=north] at (tR1) {$3,$};
  \node[anchor=north] at (tR2) {$3,$};

  \node[anchor=north] at (tL) {$2,$};
  \node[anchor=north] at (tR) {$2,$};

  \node[anchor=north] at (tRoot) {$1,$};

\end{tikzpicture}
}
\end{tabular}
}
\caption{Perfect Binary Tree for steps $\in\{1,2,3,4,5\}$ and using examples.}
\label{fig:PBT_steps}
\end{figure}

If we now use a sequence instead of the sets $V$ and $E$, it will help us to implicitly take into account the topological structure of the tree. Instead of vertices, let's use their height $h$ from $1$ for above to $n-1$ for below. Then we simply "project" the resulting numbers in the diagram into a sequence, implicitly encoding edges, see Fig. \ref{fig:PBT_steps}.
In terms of tensor space indices, this can be more easily rewritten as a sequence of level indices from left to right:
\begin{equation}
a_{\cdot,n} = (a_{\cdot,n-1} + 1) \circ \{1\} \circ (a_{\cdot,n-1} + 1), \quad a_{\cdot,2} = \{1\}
\end{equation}

In such a sequence, by construction, all indeces occur an even number of times, with the exception of index $1$. We append it to the end of the sequence and obtain:
\begin{equation}
A_{\cdot,n} =  a_{\cdot,n} \circ \{1\}
\end{equation}

Since $A_{n}$ and PBT are in one-to-one correspondence, we will also call $A_n$ a binary tree.
Let's prove that such a sequence $A_n$ meets the necessary conditions \ref{sec:NCOTAN}.

\begin{prop}[The binary tree with an additional element $\{1\}$ satisfies the necessary conditions for the Section \ref{sec:NCOTAN}.]
\end{prop}
\begin{proof}

The binary tree with concatenation by $\{1\}$ is recursively defined as:

\begin{equation}
A_{\cdot,n} =  a_{\cdot,n} \circ \{1\} , \quad a_{\cdot,n} = (a_{\cdot,n-1} + 1) \circ \{1\} \circ (a_{\cdot,n-1} + 1), \quad a_{\cdot,2} = \{1\}
\label{fractal_sequence}
\end{equation}
\begin{enumerate}
    \item The length of the sequence $|A_n| = |a_n|+1$. Since $|a_n| = 2|a_{n-1}|+1$ and $|a_2| = 1$, we obtain
    \begin{equation}
        |a_n| = \underbrace{1+2(1+2(1+\dots+2(1+1)))}_{2^{n-1}} = 1+2+4+\dots+2^{n-2} = \frac{2^{n-1}-1}{2-1} = 2^{n-1}-1
    \end{equation}
    Then:
    \begin{equation}
         |A_n| = |a_n|+1 = 2^{n-1}-1+1 = 2^{n-1}
    \end{equation}
    \item The count of identical elements in the binary tree starting from 2 is even by construction. For the tree node $1$, its pair is the element from the concatenated singleton sequence $\{1\}$.
\end{enumerate}
\end{proof}

Let us solve the problem with our ansatz using such a form of the sequence $A_n$ \eqref{fractal_sequence}.

Then we can rewrite the formula \eqref{eq:U_tail} using the fact that the sequence $A_n$ contains all the numbers from the set $\{1,2,\dots,n-1\}$. Thus, we get the product of all permutations by simply opening the brackets and giving similar terms \eqref{eq:sumprodprod}.

\begin{equation}
    U_{tail} = \sum_{\rho\in \{0,1\}^{n-1}} \bigotimes_{j = 1}^{n-1} \pi_{\rho_{j}}\prod_{i=1}^{2^{n-1}} D_1(\beta_i)X^{\chi(A_n,\rho)_i}
    \label{eq:sumprodprod}
\end{equation}

Now the entire dependence on the sequence $A_n$ is encoded via the mapping $\chi$ to a sequence of the same length: $\chi(A_n,\rho)$.

It is also worth noting that we will not construct the exact mapping $\chi$ in this section, since we are not so much interested in the mapping itself as in its properties: because the number of $X$ operators after multiplication in each monomial is preserved, and since it was even, it remains so. Therefore, the main property of the function $\chi$ is:

\begin{equation}
    \sum_{i=1}^{2^{n-1}}\chi(A_n,\rho)_i\mod2 = 0
\end{equation}

Then:
\begin{equation}
    \exists \epsilon, \delta: \prod_{i=1}^{2^{n-1}} D_1(\beta_i)X^{\chi(A_n, \rho)_i} = D_1(\epsilon,\delta)
\end{equation}

Moreover, it is possible to write relations for $\epsilon$ and $\delta$:

\begin{equation}
    X\in SU(2), \forall i D_1(\beta_i) \in SU(2) \Rightarrow \prod_{i=1}^{2^{n-1}} D_1(\beta_i)X^{\chi(A_n,\rho)_i} \in SU(2) \Rightarrow \epsilon = -\delta \Rightarrow D_1(\epsilon,\delta) = D_1(\epsilon)
\end{equation}

Recalling the notation \eqref{eq:U_tail_general}, we obtain that $\forall i: \gamma'_i = 0$.

Then we can write:
\begin{equation}
\label{eq:U_tail_block_diagonal}
U_{tail} = \begin{bmatrix}
D_1(\gamma_1) & & & & &  \\
& D_1(\gamma_2) & & & & \\
& & D_1(\gamma_3) & & & \\
& & & \dotsc & & \\
& & & & & & D_1(\gamma_{2^{n-1}-1}) & \\
& & & & & & & D_1(\gamma_{2^{n-1}}) \\
\end{bmatrix}
\end{equation}

Now, let's define map between $\gamma$ and $\beta$:

\begin{definition}[$r_n$ matrix]
\label{def:r}
$r_n$ is the matrix of the linear mapping between the parameters $\beta$ in the decomposition \eqref{eq:sumprodprod} and the parameters $\gamma$ in the block diagonal form \eqref{eq:U_tail_block_diagonal}:
    \begin{equation}
\gamma= \beta r_n
\label{eq:linear_formula}
    \end{equation}
or with matrix indices $i$ adn $j$:
    \begin{equation}
    \gamma_i = \beta_j    r_{i,n}^j
    \end{equation}
\end{definition}

Formula \eqref{eq:linear_formula} is obvious since the presence of an even number of $X$ operators in the formula \eqref{eq:sumprodprod}, as already mentioned, can only change the sign of individual summation parameters. Leaving aside the explicit dependence at this stage, it can be understood that the dependence is at least linear, which means that the map $\mathcal{L}$ is also linear.

\section{Structure of $r_n$
\label{sec:r_formula}}

Now let us delve deeper into the structure of $r_n$, defined by Def. \ref{def:r}. It is important to understand that this is the matrix of a linear mapping between the parameters $\vec{\beta}$ and $\vec{\gamma}$.

Since the number of $X$ operators is even for each $\rho$, each pair of $X$ operators flips the sign of the parameters of all diagonal matrices between them (see \eqref{eq:D_prop}), we can perform the following transformation:

\begin{equation}
    \prod_{i=1}^{2^{n-1}} D_1(\beta_i)X^{\chi(A_n,\rho)_i} = \prod_{i=1}^{2^{n-1}} D_1((-1)^{\sum_{k = 1}^{i-1}\chi(A_n,\rho)_k}\beta_i) = D_1\bigg(\sum_{i=1}^{2^{n-1}}(-1)^{\sum_{k = 1}^{i-1}\chi(A_n,\rho)_k}\beta_i \bigg)
\end{equation}

Since $\rho$ is simply the binary representation of the index of the matrix in the block-diagonal matrix $U_{tail}$, then, considering:

\begin{equation}
    j(\rho) = \sum_{i=1}^{{n-1}}2^{n-1-i}\rho_{i}, \quad \rho_i = \frac{j\mod 2^{n-i} - (j\mod2^{n+1-i}) }{2^{n-1-i}}
\end{equation}

The proof is given in Appendix \ref{sec:j_to_rho}.
Then we can write several transformations:

\begin{align}
    D_1(\gamma_j) = D_1\bigg(\sum_{i=1}^{2^{n-1}}(-1)^{\sum_{k = 1}^{i-1}\chi(A_n,\rho(j))_k}\beta_i \bigg) \Rightarrow \\
    \gamma_j = \sum_{i=1}^{2^{n-1}}(-1)^{\sum_{k = 1}^{i-1}\chi(A_n,\rho(j))_k}\beta_i \Rightarrow\\
    \gamma_j = r_{j,n}^i\beta_i, \quad r_{j,n}^i = (-1)^{\sum_{k = 1}^{i-1}\chi(A_n,\rho(j))_k}
\end{align}

So, we have found the linear mapping, we see its form, and now all that remains is to finally understand the structure of the mapping \(\chi\).
The columns of this matrix depend on \(i\) as follows:

\begin{align}
    r^1_{\cdot,n} = \begin{bmatrix}
        1\\1\\\cdots\\1
    \end{bmatrix} \quad r^2_{\cdot,n} = \begin{bmatrix}
        (-1)^{\chi(A_n,\rho(1))_1}\\
        (-1)^{\chi(A_n,\rho(1))_1}\\
        \cdots \\
        (-1)^{\chi(A_n,\rho(2^{n-1}))_1}
\end{bmatrix} \quad \cdots\quad r^i_{\cdot,n} = \begin{bmatrix}
        (-1)^{\sum_{k = 1}^{i-1}\chi(A_n,\rho(1))_k}\\
        (-1)^{\sum_{k = 1}^{i-1}\chi(A_n,\rho(2))_k}\\
        \cdots \\
        (-1)^{\sum_{k = 1}^{i-1}\chi(A_n,\rho(2^{n-1}))_k}\\
\end{bmatrix} \\ r^{i+1}_{\cdot,n} = \begin{bmatrix}
        (-1)^{\sum_{k = 1}^{i}\chi(A_n,\rho(1))_k}\\
        (-1)^{\sum_{k = 1}^{i}\chi(A_n,\rho(2))_k}\\
        \cdots \\
        (-1)^{\sum_{k = 1}^{i}\chi(A_n,\rho(2^{n-1}))_k}\\
\end{bmatrix}
\end{align}

Consequently, it is clear that neighboring columns differ by element-wise multiplication by the column:

\begin{equation}
    \bar{r}^i_{\cdot,n} = \begin{bmatrix}
        (-1)^{\chi(A_n,\rho(1))_i}\\
        (-1)^{\chi(A_n,\rho(2))_i}\\
        \cdots \\
        (-1)^{\chi(A_n,\rho(2^{n-1}))_i}\\
    \end{bmatrix}
\end{equation}

But what exactly does this form mean?

It precisely corresponds to multiplying the angle parameters by $1$ or $-1$ depending on the presence or absence of the operator $X$ at the previous step, according to formula \eqref{eq:U_tail}. Indeed, in the $k$-th tensor product, the angle formula then splits as follows: $\pi_0 \rightarrow 1$, $\pi_1 \rightarrow -1$, and as a result, such a column looks like:

\begin{equation}
    \bar{r}^i_{\cdot,n} = ([1,1]^{\ot k(i)-1}\ot[1,-1]\ot[1,1]^{\ot n-1-k(i)})^T.
\end{equation}

But in this case, $k(i)$ is simply an element of the sequence $A_n$, then

\begin{equation}
    \bar{r}^i_{\cdot,n} = ([1,1]^{\ot A_n(i)-1}\ot[1,-1]\ot[1,1]^{\ot n-1-A_n(i)})^T.
\end{equation}

Then in the end the matrix looks like
\begin{equation}
r_n = \left[\prod_{k=1}^i \bar{r}^i_{\cdot,n} \text{ for } i\in\{1,\dots,2^{n-1}\}\right],
\end{equation}

where $\prod$ denotes element by element multiplication.

This is the matrix of the linear transformation obtained during the recursive construction with the given rule for $A_n$, which we were looking for.

For $n\in\{2, 3, 4\}$, the matrices look as follows, and for the sequences $A_n$ defined in order \eqref{fractal_sequence}, we get the matrices:

\begin{equation}
r_2 = \begin{bmatrix}1 & 1 \\ 1 & -1\end{bmatrix}
r_3 = \begin{bmatrix}1 & 1 & 1 & 1 \\ 1 & -1 & -1 & 1 \\ 1 & 1 & -1 & -1 \\ 1 & -1 & 1 & -1\end{bmatrix}
r_4 = \begin{bmatrix} 1 & 1 & 1 & 1 & 1 & 1 & 1 & 1\\
1 & -1 & -1 & 1 & 1 & -1 & -1 & 1 \\
1 & 1 & -1 & -1 & -1 & -1 & 1 & 1 \\
1 & -1 & 1 & -1 & -1 & 1 & -1 & 1 \\
1 & 1 & 1 & 1 & -1 & -1 & -1 & -1 \\
1 & -1 & -1 & 1 & -1 & 1 & 1 & -1 \\
1 & 1 & -1 & -1 & 1 & 1 & -1 & -1 \\
1 & -1 & 1 & -1 & 1 & -1 & 1 & -1\end{bmatrix}
\label{eq:matrices}
\end{equation}

However, in all rows and columns, except for the bottom row and the left column, the number of $+1$ and $-1$ match. This suggests that matrices can be created by rearranging the runoff and columns of the tensor product of the very first matrix, as can be seen numerically from examples.

For each of these three matrices, you can make sure that their inverses obey the relation:

\begin{equation}
    r_n^{-1} = \frac{1}{2^{n-1}}r_n^T
    \label{eq:r_n_minus1}
\end{equation}

The relationship between the absolute values of the determinants for different orders $n$ is given by:

\begin{equation}
    |\det(r_n)| = |\det(r_2)|^{(n-1)\cdot2^{n-2}}
\end{equation}

Note that the dependence of the tensor degree $r_2$ is as follows:

\begin{equation}
    |\det(r_2)^{\ot k}| = |\det(r_2)|^{k\cdot2^{k-1}}
\end{equation}

That is, just as if the equality were true up to the permutation of rows and columns:

\begin{equation}
    r_n =_\sim r_2^{\ot (n-1)}
    \label{eq:r_nrecursive}
\end{equation}

where $=_\sim$ means "up to permutations". In the next Section \ref{sec:rT}, we will prove this statement.

\subsection{Reverse $\mathcal{L}$ mapping\label{sec:rT}}

For finding $\mathcal{L}$ we firstly should invert $r_n$ matrix.

For $n=2$: $r_2 = \begin{bmatrix}1&1\\1&-1\end{bmatrix}$

\begin{theorem}{About permutations.\label{th:core}}
    Let $\sigma \in S_{2^n}$ and $\sigma(r_n) = (r^{\sigma(1)}_{\cdot,n}, r^{\sigma(2)}_{\cdot,n}, \dots, r^{\sigma(2^n)}_{\cdot,n})$, where $r^i_{\cdot,n}$ is the $i$-th column of the matrix $r_n$, then
\begin{equation}
    \forall n \in \mathbb{N}\ \exists \sigma:\ r_{n+1} = \sigma(r_2^{\otimes n})
\end{equation}

\end{theorem}

\begin{proof}
\begin{equation}     r^i_{\cdot,n+1} = \prod_{k=1}^i r_{A_{n+1}(k),n+1} = \prod_{k=1}^i \begin{bmatrix}1\\1\end{bmatrix}^{\ot A_{n+1}(k)-1}\ot\begin{bmatrix}1\\-1\end{bmatrix}\ot\begin{bmatrix}1\\1\end{bmatrix}^{\ot n-A_{n+1}(k)} = \end{equation}
\begin{equation}     = \bigotimes_{m=1}^n\begin{bmatrix}1\\-1\end{bmatrix}^{\Sigma_{k=1}^i\delta_{m,A_{n+1}(k)}}
\label{eq:general_from_core}\end{equation}

On the other hand, $r_2^{\ot n}$ has columns of the form:\begin{equation}     \bigotimes_{j=1}^n\begin{bmatrix}1\\-1\end{bmatrix}^{p_j},\quad \{p_j\}\in\{0,1\}^n \end{equation}

The values of $\sum_{k=1}^i\delta_{m,A_{n}(k)}$ are of interest only modulo $2$, because $\begin{bmatrix}1\\-1\end{bmatrix}^{2} = \begin{bmatrix}1\\1\end{bmatrix}$ in the sense of element by element product.

The proof ends with an analysis of the diagram (Fig. \ref{fig:visual_proof}). It shows the values of the expressions $\Sigma_{k=1}^i\delta_{m,A_{n}(k)}\mod 2$ for each $i$ and $m\in {1\dots n}$. This means that vertically we get all the combinations of zeros and ones from the set $\{0,1\}^n$, but this means that all the columns of the matrix $r_n$ coincide with the columns of tensor degree $r_2$.
\end{proof}

It is worth clarifying exactly how the diagram is constructed: each row is assigned to each subspace numbered by $m\in\{1,\dots,n-1\}$. The row is a continuous line with two values indicated by horizontal lines and vertical lines where the values change. We call the vertical lines 'gaps'. The horizontal line running above corresponds to  $\Sigma_{k=1}^i\delta_{m,A_{n}(k)}\mod 2 = 1$ case, and the horizontal line running below corresponds to $\Sigma_{k=1}^i\delta_{m,A_{n}(k)}\mod 2 = 0$ case. The number $i$ corresponding to the upper limit of the sum is the abscissa of the representation image of the diagram.

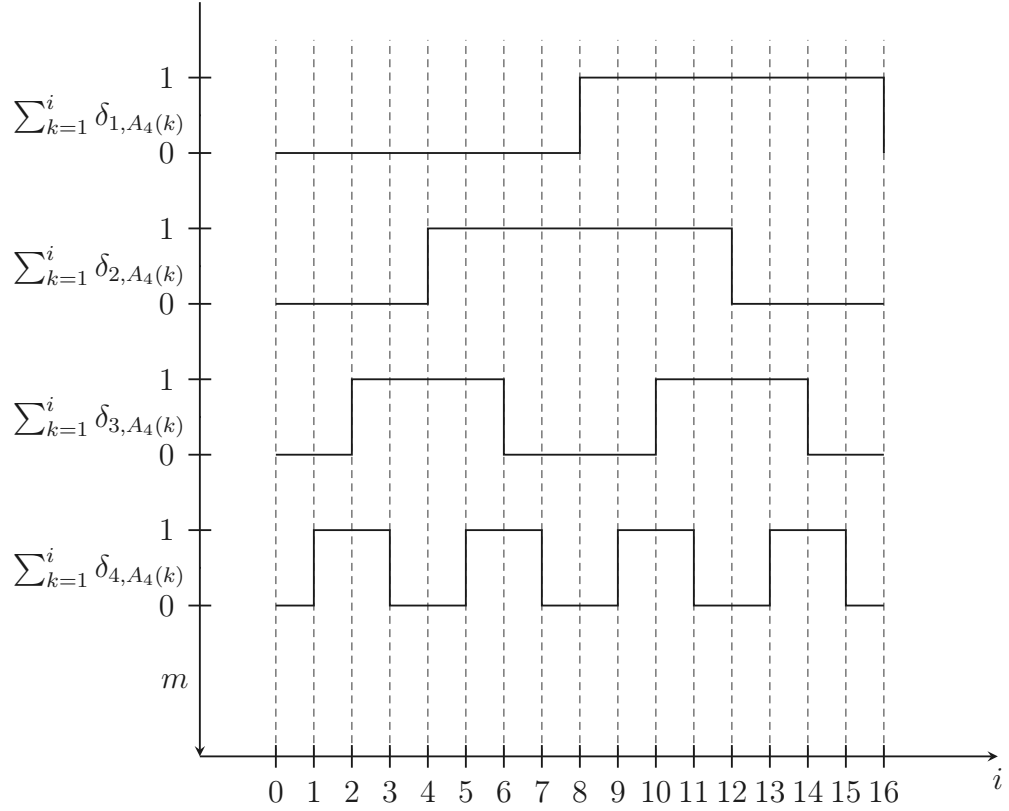
\begin{figure}

    \centering
    \begin{tikzpicture}[
    every path/.style={color=black!90, line width=0.7pt},
    every node/.style={color=black!90}
  ]
  \def\Xstart{1}    
  \def\Xspan{8}     
  \def\Xend{\Xstart+\Xspan} 
  \def\Ymax{10}     
  \def\Amp{0.5}     

  \tikzstyle{grid style}=[densely dashed, color=black!60, line width=0.5pt]

  \draw[-stealth] (0,0) -- (\Xend+1.5, 0) node[below] {$i$};
  \draw[-stealth] (0, \Ymax) -- (0, 0);

  \node[left] at (0, 1){$m$};


\draw[line width=0.7pt] (0, 8.5) -- (-0, 8.5) node[left=1pt] {$\sum_{k=1}^i\delta_{1,A_4(k)}$};
  \draw[line width=0.7pt] (0.15, 9) -- (-0.15, 9) node[left=1pt] {$1$};
  \draw[line width=0.7pt] (0.15, 8) -- (-0.15, 8) node[left=1pt] {$0$};
  \draw[line width=0.7pt] (0, 6.5) -- (-0, 6.5) node[left=1pt] {$\sum_{k=1}^i\delta_{2,A_4(k)}$};
  \draw[line width=0.7pt] (0.15, 7) -- (-0.15, 7) node[left=1pt] {$1$};
  \draw[line width=0.7pt] (0.15, 6) -- (-0.15, 6) node[left=1pt] {$0$};
  \draw[line width=0.7pt] (0, 4.5) -- (-0, 4.5) node[left=1pt] {$\sum_{k=1}^i\delta_{3,A_4(k)}$};
  \draw[line width=0.7pt] (0.15, 5) -- (-0.15, 5) node[left=1pt] {$1$};
  \draw[line width=0.7pt] (0.15, 4) -- (-0.15, 4) node[left=1pt] {$0$};
  \draw[line width=0.7pt] (0, 2.5) -- (-0, 2.5) node[left=1pt] {$\sum_{k=1}^i\delta_{4,A_4(k)}$};
  \draw[line width=0.7pt] (0.15, 3) -- (-0.15, 3) node[left=1pt] {$1$};
  \draw[line width=0.7pt] (0.15, 2) -- (-0.15, 2) node[left=1pt] {$0$};


  \foreach \x [evaluate=\x as \xx using {(\x+2)/2}] in {0,...,16} {
    \draw[grid style] (\xx, 0) -- ++(0, 9.5);
    \draw[line width=0.7pt] (\xx, 0.15) -- ++(0, -0.3)
      node[below=1pt] {\x};
  }


  \def\Y{8.5} 
  \draw (\Xstart, \Y-\Amp) -- (\Xstart + \Xspan/2, \Y-\Amp) |- (\Xstart+ \Xspan/2, \Y+\Amp) |- (\Xend, \Y+\Amp)|- (\Xend, \Y-\Amp);

  \def\Y{6.5} 
  \draw (\Xstart, \Y-\Amp) -- (\Xstart + \Xspan/4, \Y-\Amp) |- (\Xstart + 3*\Xspan/4, \Y+\Amp) |- (\Xend, \Y-\Amp);

  \def\Y{4.5} 
  \draw (\Xstart, \Y-\Amp) -- (\Xstart + \Xspan/8, \Y-\Amp) |- (\Xstart + 3*\Xspan/8, \Y+\Amp) |- (\Xstart + 5*\Xspan/8, \Y-\Amp) |- (\Xstart + 7*\Xspan/8, \Y+\Amp)|- (\Xend, \Y-\Amp);

  \def\Y{2.5} 
  \draw (\Xstart, \Y-\Amp) -- (\Xstart + 1*\Xspan/16, \Y-\Amp)
      |- (\Xstart + 3*\Xspan/16, \Y+\Amp)
      |- (\Xstart + 5*\Xspan/16, \Y-\Amp)
      |- (\Xstart + 7*\Xspan/16, \Y+\Amp)
      |- (\Xstart + 9*\Xspan/16, \Y-\Amp)
      |- (\Xstart + 11*\Xspan/16, \Y+\Amp)
      |- (\Xstart + 13*\Xspan/16, \Y-\Amp)
      |- (\Xstart + 15*\Xspan/16, \Y+\Amp)
      |- (\Xend, \Y-\Amp);

\end{tikzpicture}
    \caption{A diagram generated by a binary tree for $n = 4$. The generalization is obvious.}
    \label{fig:visual_proof}
\end{figure}

\begin{theorem}{$r_{n+1}^{-1} = \frac{1}{2^n}r_{n+1}^T$ \label{th:inverse}}
\end{theorem}
\begin{proof}Let $P_1$ and $P_2$ be row and column permutation operators, respectively.
Using the properties of tensor products $(r_2^{\ot n})^{-1} =(r_2^{-1})^{\ot n}$, $(r_2^{\ot n})^{T}=(r_2^{T})^{\ot n}$ and the fact that $r_2^{-1} = \frac{1}{2}r_2^T$, we get:

        \begin{multline}
         (r_2^{\ot n})^{-1} = (P_1 r_{n+1} P_2)^{-1} = P_2^{-1}r_{n+1}^{-1}P_1^{-1} = \frac{1}{2^n}(r_2^{\ot n})^T \xrightarrow{}\\ \xrightarrow{}r_{n+1}^{-1} = \frac{1}{2^n}P_2(r_2^{\ot n})^TP_1 =\xrightarrow[]{(P_1r_{n+1}P_2)^T = P_2^Tr_{n+1}^TP_1^T = (r_2^{\ot n})^T}\\ = \frac{1}{2^n} P_2P_2^Tr_{n+1}^TP_1^TP_1 =\xrightarrow[]{P_iP_i^T = I_{2^n}}= \frac{1}{2^n} r_{n+1}^T
         \end{multline}

        That was what had to be proved.
\end{proof}

\subsection{Full proof of the Theorem \ref{th:reccurent}\label{sec:fullproof}}

Let's put together all the information discussed in this article to prove Th. \ref{th:reccurent}.

According to section \ref{sec:Chapter5}, if the conditions on $U_{tail}$ discussed in section \ref{sec:CNOTSCHEMEOUR} and the conditions on $A_n$ discussed in section \ref{sec:NCOTAN} are satisfied, the matrix on the right-hand side of the factorization \eqref{eq:certralreccurentformula} is indeed diagonal. According to \ref{sec:linearityL}, this defines a linear mapping $\mathcal{L}$ between the parameters $\vec\beta$ and $\vec{\bar\alpha}$ and the parameters of the diagonal matrix $\vec\alpha$ that corresponds to them:

\begin{align}
\mathcal{L} =
\begin{cases}
\alpha_{2i-1} &= \bar\alpha_i + \beta_jr^j_{i,n} \\
\alpha_{2i} &= \bar\alpha_i - \beta_jr^j_{i,n} \\
\end{cases}
\end{align}

Or in vector form:

\begin{align}
\mathcal{L}: \vec\alpha = \vec{\bar{\alpha}}\ot[1,1] + \big(\vec\beta \cdot r_n \big) \ot [1,-1]
\end{align}

According to Section \ref{sec:NCOTAN}, this mapping becomes nonsingular and invertible provided that the matrix $r_n$ is nonsingular and invertible. According to Th.~\ref{th:core}, $r_n$ is nonsingular if and only if $A_n$ is spanned by a diagram that contains all combinations of $0$ and $1$, for example, the binary tree sequence (see Def. \ref{def:PBT}). Under these conditions, according to Th.~\ref{th:inverse}, the inverse matrix is given by the formula $r_n^{-1} = \frac{1}{2^{n-1}}r^T_{n}$, then the inverse mapping $\mathcal{L}^{-1}$ is given by the formula:

\begin{align}
\mathcal{L}^{-1} =
\begin{cases}
\bar\alpha_i &= \frac{\alpha_{2i-1}+\alpha_{2i}}{2} \\
\beta_i &= \frac{\alpha_{2i-1}+\alpha_{2i}}{2^n}r^i_{j,n} \\
\end{cases}
\end{align}

Or, in vector form, let $\vec A_1 = [\alpha_1,\alpha_3,\dots,\alpha_{2^{n-1}}]$, $\vec A_2 = [\alpha_2, \alpha_4,\dots,\alpha_{2^{n}}]$ be even and odd parameters of $D_n$, then:

\begin{align}
\mathcal{L}^{-1}: \quad \vec{\bar{\alpha}} = \frac{\vec A_1+\vec A_2}{2} \qquad
\vec{\beta} = \frac{1}{2^n}(\vec A_1-\vec A_2)r^T
\end{align}

Thus, we have found an invertible bijective mapping between the parameters of the quantum circuit and the parameters of the diagonal matrix.

This completes the proof of Th. \ref{th:reccurent}, given at the beginning of this paper.

\section{Symmetries of $A_n$\label{sec:symmetries}}

As mentioned above in the Section \ref{sec:rT}, all diagrams are in one-to-one correspondence with the $A_n$ sequences. A necessary and sufficient condition for a diagram corresponding to $A_n$ is the presence in it of all possible sequences of $\{0,1\}$ of length $n$, which means that all possible diagrams can be generated by permutation the rows and columns of the diagram in Fig. \ref{fig:visual_proof}. We can rearrange the rows of this diagram - that is, the plots of the sums $\sum_{k=1}^i\delta_{m,A_n(k)}$ - this will simply be a renumbering of tensor spaces with $\{1,2,\dots,n-1\}$, as shown in Fig. \ref{fig:diag_row_comm}. Changing a chart by permutation columns requires a deeper explanation. In Fig. \ref{fig:diag_column_comm} it can be seen that in the general case, permutation the columns can change not only the places of the 'gaps', but also their number. As we have already discussed, the gap in the $k-$th row with the coordinate $i$ corresponds to the presence of the operator $L(m)$ between $I^{\ot(n-1)}\ot D_1(\beta_i)$ and $I^{\ot(n-1)}\ot D_1(\beta_{i+1})$ in the formula \eqref{eq:certralreccurentformula}. The schemes resulting from such permutations always turn out to be non-degenerate, since the Th. \ref{th:core} is still applicable to them, since only the fact that the corresponding diagram contains all combinations of sequences of zeros and ones of length $n-1$, \eqref{eq:general_from_core}.

Let's take a closer look at two specific examples, in one of which the total number of operators $L(k)$ remains the same and an example in which their number increases. Cases in which the number of 'gaps' decreases are not possible. The explanation for this can be given as follows: at each next level (next rows), two options with $0$ and $1$ should be added for each combination. These values will be separated by at least one 'gap'. In this case, each level will have at least twice as many 'gaps' as the previous one.  This is the lower bound, and therefore, if the binary tree case satisfies it, then this is the minimal case.

\begin{figure}[htbp]
    \centering

\begin{minipage}{0.3\linewidth}
    \centering

    \begin{minipage}{0.3\linewidth}
        \centering
        \begin{tikzpicture}[
                    scale=0.5,
                    every path/.style={color=black!90, line width=0.7pt},
                    every node/.style={color=black!90}
                  ]
                  \def\Xstart{1}    
                  \def\Xspan{8}     
                  \def\Xend{\Xstart+\Xspan} 
                  \def\Ymax{10}     
                  \def\Amp{0.5}     

                  \tikzstyle{grid style}=[densely dashed, color=black!60, line width=0.5pt]

                \foreach \x [evaluate=\x as \xx using {(\x+2)/2}] in {0,...,16} {
    \draw[grid style] ({\xx}, 1.5) -- ++(0, 8);

  }
\draw[grid style, red!90] (2.5, 1.5) -- ++(0, 8);
\draw[grid style, red!90] (3, 1.5) -- ++(0, 8);
  \draw[grid style, red!90] (5, 1.5) -- ++(0, 8);
\draw[grid style, red!90] (5.5, 1.5) -- ++(0, 8);


  \def\Y{8.5} 
  \draw (\Xstart, \Y-\Amp) -- (\Xstart, \Y-\Amp) |- (\Xstart + \Xspan/2, \Y-\Amp) |- (\Xend, \Y+\Amp);

  \def\Y{6.5} 
  \draw (\Xstart, \Y+\Amp) -- (\Xstart + \Xspan/4, \Y+\Amp) |- (\Xstart + 3*\Xspan/4, \Y-\Amp) |- (\Xend, \Y+\Amp);

  \def\Y{4.5} 
  \draw (\Xstart, \Y+\Amp) -- (\Xstart + \Xspan/8, \Y+\Amp) |- (\Xstart + 3*\Xspan/8, \Y-\Amp) |- (\Xstart + 5*\Xspan/8, \Y+\Amp) |- (\Xstart + 7*\Xspan/8, \Y-\Amp)|- (\Xend, \Y+\Amp);

  \def\Y{2.5} 
  \draw (\Xstart, \Y+\Amp) -- (\Xstart + 1*\Xspan/16, \Y+\Amp)
      |- (\Xstart + 3*\Xspan/16, \Y-\Amp)
      |- (\Xstart + 5*\Xspan/16, \Y+\Amp)
      |- (\Xstart + 7*\Xspan/16, \Y-\Amp)
      |- (\Xstart + 9*\Xspan/16, \Y+\Amp)
      |- (\Xstart + 11*\Xspan/16, \Y-\Amp)
      |- (\Xstart + 13*\Xspan/16, \Y+\Amp)
      |- (\Xstart + 15*\Xspan/16, \Y-\Amp)
      |- (\Xend, \Y+\Amp);

                \end{tikzpicture}
    \end{minipage}

    \begin{minipage}{.3\linewidth}

                \begin{tikzpicture}[scale = .6]
                \node (A) at (1.0,1) { };
                \node (B) at (3.5,0) { };
                \draw[->] (A) -- (B) node[midway, above] {};
                \node (C) at (3.5,1) { };
                \node (E) at (0,0) { };
                \node (D) at (1,0) { };
                \draw[->] (C) -- (D) node[midway, above] {};
            \end{tikzpicture}

        \centering
                    \begin{tikzpicture}[
                    scale=0.5,
                    every path/.style={color=black!90, line width=0.7pt},
                    every node/.style={color=black!90}
                  ]
                  \def\Xstart{1}    
                  \def\Xspan{8}     
                  \def\Xend{\Xstart+\Xspan} 
                  \def\Ymax{10}     
                  \def\Amp{0.5}     

                  \tikzstyle{grid style}=[densely dashed, color=black!60, line width=0.5pt]

                \foreach \x [evaluate=\x as \xx using {(\x+2)/2}] in {0,...,16} {
    \draw[grid style] ({\xx}, 1.5) -- ++(0, 8);

  }

                \draw[grid style, red!90] (2.5, 1.5) -- ++(0, 8);
\draw[grid style, red!90] (3, 1.5) -- ++(0, 8);
  \draw[grid style, red!90] (5, 1.5) -- ++(0, 8);
\draw[grid style, red!90] (5.5, 1.5) -- ++(0, 8);

  \def\Y{8.5} 
  \draw (\Xstart, \Y-\Amp) -- (\Xstart + 3*\Xspan/16, \Y-\Amp) |-(\Xstart + 4*\Xspan/16, \Y+\Amp) |-(\Xstart + \Xspan/2+\Xspan/16, \Y-\Amp) |- (\Xend, \Y+\Amp);

  \def\Y{6.5} 
  \draw (\Xstart, \Y+\Amp) -- (\Xstart + \Xspan/4 -\Xspan/16, \Y+\Amp) |- (\Xstart + 8*\Xspan/16, \Y-\Amp) |- (\Xstart + 9*\Xspan/16, \Y+\Amp)|- (\Xstart + 3*\Xspan/4, \Y-\Amp) |- (\Xend, \Y+\Amp);

  \def\Y{4.5} 
  \draw (\Xstart, \Y+\Amp) -- (\Xstart + \Xspan/8, \Y+\Amp) |- (\Xstart + 3*\Xspan/16, \Y-\Amp) |-
  (\Xstart + 4*\Xspan/16, \Y+\Amp) |- (\Xstart + 6*\Xspan/16, \Y-\Amp)
  |- (\Xstart + \Xspan/2, \Y+\Amp)
  |- (\Xstart + \Xspan/2+\Xspan/16, \Y-\Amp)
  |- (\Xstart + \Xspan/2+\Xspan/8, \Y+\Amp)
  |- (\Xstart + 7*\Xspan/8, \Y-\Amp)|- (\Xend, \Y+\Amp);

  \def\Y{2.5} 
  \draw (\Xstart, \Y+\Amp) -- (\Xstart + 1*\Xspan/16, \Y+\Amp)
      |- (\Xstart + 3*\Xspan/16, \Y-\Amp)
      |- (\Xstart + 5*\Xspan/16, \Y+\Amp)
      |- (\Xstart + 7*\Xspan/16, \Y-\Amp)
      |- (\Xstart + 9*\Xspan/16, \Y+\Amp)
      |- (\Xstart + 11*\Xspan/16, \Y-\Amp)
      |- (\Xstart + 13*\Xspan/16, \Y+\Amp)
      |- (\Xstart + 15*\Xspan/16, \Y-\Amp)
      |- (\Xend, \Y+\Amp);

                \end{tikzpicture}

    \end{minipage}
\caption{Permutation the diagram columns.}
        \label{fig:diag_column_comm}


\end{minipage}
\hfill
\begin{minipage}{0.6\linewidth}
    \centering
    \begin{minipage}{0.44\linewidth}
        \centering
                    \begin{tikzpicture}[
                    scale=0.45,
                    every path/.style={color=black!90, line width=0.7pt},
                    every node/.style={color=black!90}
                  ]
                  \def\Xstart{1}    
                  \def\Xspan{8}     
                  \def\Xend{\Xstart+\Xspan} 
                  \def\Ymax{10}     
                  \def\Amp{0.5}     

                  \tikzstyle{grid style}=[densely dashed, color=black!60, line width=0.5pt]

                  \foreach \x [evaluate=\x as \xx using {(\x+2)/2}] in {0,...,16} {
    \draw[grid style] ({\xx}, 1.5) -- ++(0, 8);

  }

  \def\Y{8.5} 
  \draw (\Xstart, \Y-\Amp) -- (\Xstart + \Xspan/2, \Y-\Amp) |- (\Xend, \Y+\Amp);

  \def\Y{6.5} 
  \draw[red!90] (\Xstart, \Y+\Amp) -- (\Xstart + \Xspan/4, \Y+\Amp) |- (\Xstart + 3*\Xspan/4, \Y-\Amp) |- (\Xend, \Y+\Amp);

  \def\Y{4.5} 
  \draw (\Xstart, \Y+\Amp) -- (\Xstart + \Xspan/8, \Y+\Amp) |- (\Xstart + 3*\Xspan/8, \Y-\Amp) |- (\Xstart + 5*\Xspan/8, \Y+\Amp) |- (\Xstart + 7*\Xspan/8, \Y-\Amp)|- (\Xend, \Y+\Amp);

  \def\Y{2.5} 
  \draw[red!90] (\Xstart, \Y+\Amp) -- (\Xstart + 1*\Xspan/16, \Y+\Amp)
      |- (\Xstart + 3*\Xspan/16, \Y-\Amp)
      |- (\Xstart + 5*\Xspan/16, \Y+\Amp)
      |- (\Xstart + 7*\Xspan/16, \Y-\Amp)
      |- (\Xstart + 9*\Xspan/16, \Y+\Amp)
      |- (\Xstart + 11*\Xspan/16, \Y-\Amp)
      |- (\Xstart + 13*\Xspan/16, \Y+\Amp)
      |- (\Xstart + 15*\Xspan/16, \Y-\Amp)
      |- (\Xend, \Y+\Amp);

                \end{tikzpicture}

    \end{minipage}
\begin{minipage}{.05\linewidth}
\raisebox{-30mm}{   \begin{tikzpicture}[scale = .05]
    \node (A) at (0,16) { };
    \node (B) at (13,-24) { };
    \draw[->] (A) -- (B) node[midway, above] {};
    \node (C) at (13,16) { };
    \node (D) at (0,-24) { };
    \node (E) at (0,0) { };
    \draw[->] (D) -- (C) node[midway, above] {};

\end{tikzpicture}
}
\end{minipage}
    \hfill
    \begin{minipage}{0.44\linewidth}
        \centering
                    \begin{tikzpicture}[
                    scale=0.45,
                    every path/.style={color=black!90, line width=0.7pt},
                    every node/.style={color=black!90}
                  ]
                  \def\Xstart{1}    
                  \def\Xspan{8}     
                  \def\Xend{\Xstart+\Xspan} 
                  \def\Ymax{10}     
                  \def\Amp{0.5}     

                  \tikzstyle{grid style}=[densely dashed, color=black!60, line width=0.5pt]

                \foreach \x [evaluate=\x as \xx using {(\x+2)/2}] in {0,...,16} {
    \draw[grid style] ({\xx}, 1.5) -- ++(0, 8);

  }


  \def\Y{8.5} 
  \draw (\Xstart, \Y-\Amp) -- (\Xstart + \Xspan/2, \Y-\Amp) |- (\Xend, \Y+\Amp);

  \def\Y{2.5} 
  \draw[red!90] (\Xstart, \Y+\Amp) -- (\Xstart + \Xspan/4, \Y+\Amp) |- (\Xstart + 3*\Xspan/4, \Y-\Amp) |- (\Xend, \Y+\Amp);

  \def\Y{4.5} 
  \draw (\Xstart, \Y+\Amp) -- (\Xstart + \Xspan/8, \Y+\Amp) |- (\Xstart + 3*\Xspan/8, \Y-\Amp) |- (\Xstart + 5*\Xspan/8, \Y+\Amp) |- (\Xstart + 7*\Xspan/8, \Y-\Amp)|- (\Xend, \Y+\Amp);

  \def\Y{6.5} 
  \draw[red!90] (\Xstart, \Y+\Amp) -- (\Xstart + 1*\Xspan/16, \Y+\Amp)
      |- (\Xstart + 3*\Xspan/16, \Y-\Amp)
      |- (\Xstart + 5*\Xspan/16, \Y+\Amp)
      |- (\Xstart + 7*\Xspan/16, \Y-\Amp)
      |- (\Xstart + 9*\Xspan/16, \Y+\Amp)
      |- (\Xstart + 11*\Xspan/16, \Y-\Amp)
      |- (\Xstart + 13*\Xspan/16, \Y+\Amp)
      |- (\Xstart + 15*\Xspan/16, \Y-\Amp)
      |- (\Xend, \Y+\Amp);

                \end{tikzpicture}

    \end{minipage}
\caption{Permutation the diagram rows.}
        \label{fig:diag_row_comm}

\end{minipage}
\end{figure}

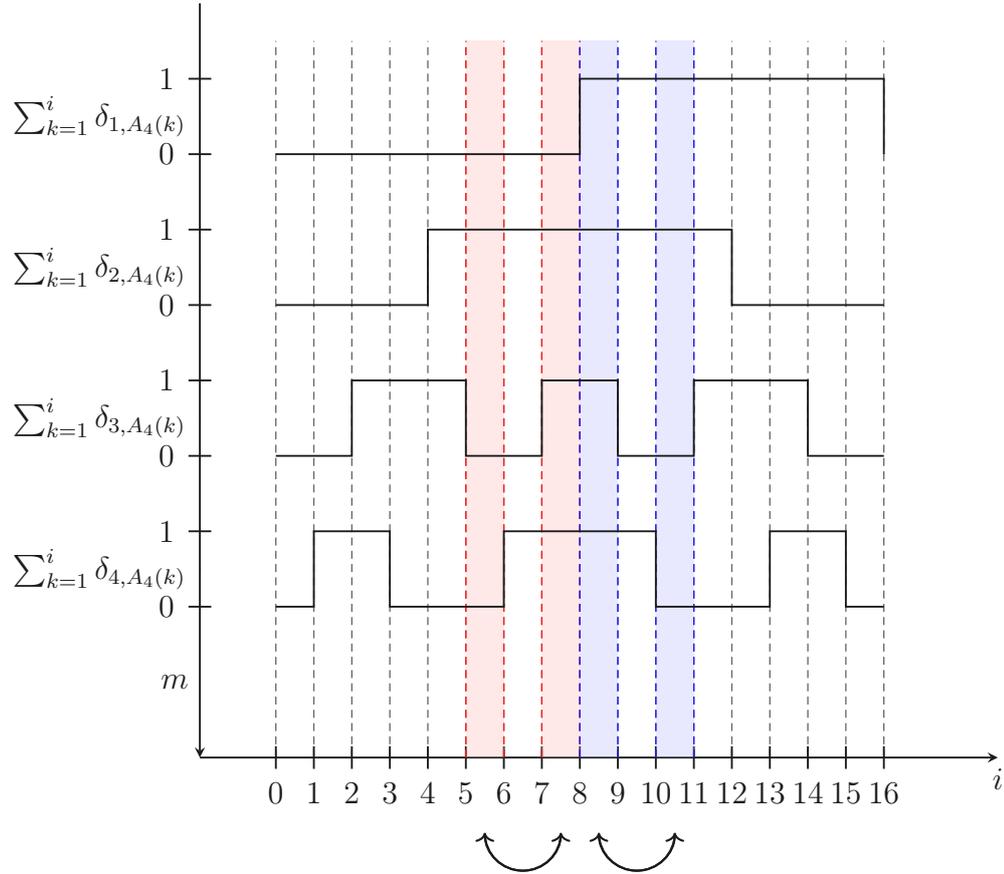
\begin{figure}

    \centering
    \begin{tikzpicture}[
    every path/.style={color=black!90, line width=0.7pt},
    every node/.style={color=black!90}
  ]
  \def\Xstart{1}    
  \def\Xspan{8}     
  \def\Xend{\Xstart+\Xspan} 
  \def\Ymax{10}     
  \def\Amp{0.5}     

  \tikzstyle{grid style}=[densely dashed, color=black!60, line width=0.5pt]

  \draw[-stealth] (0,0) -- (\Xend+1.5, 0) node[below] {$i$};
  \draw[-stealth] (0, \Ymax) -- (0, 0);

  \node[left] at (0, 1) {$m$};

  \draw[line width=0.7pt] (0, 8.5) -- (-0, 8.5) node[left=1pt] () {$\sum_{k=1}^i\delta_{1,A_4(k)}$};
  \draw[line width=0.7pt] (0.15, 9) -- (-0.15, 9) node[left=1pt] () {$1$};
  \draw[line width=0.7pt] (0.15, 8) -- (-0.15, 8) node[left=1pt] () {$0$};
  \draw[line width=0.7pt] (0, 6.5) -- (-0, 6.5) node[left=1pt] () {$\sum_{k=1}^i\delta_{2,A_4(k)}$};
  \draw[line width=0.7pt] (0.15, 7) -- (-0.15, 7) node[left=1pt] {$1$};
  \draw[line width=0.7pt] (0.15, 6) -- (-0.15, 6) node[left=1pt] {$0$};
  \draw[line width=0.7pt] (0, 4.5) -- (-0, 4.5) node[left=1pt] {$\sum_{k=1}^i\delta_{3,A_4(k)}$};
  \draw[line width=0.7pt] (0.15, 5) -- (-0.15, 5) node[left=1pt] {$1$};
  \draw[line width=0.7pt] (0.15, 4) -- (-0.15, 4) node[left=1pt] {$0$};
  \draw[line width=0.7pt] (0, 2.5) -- (-0, 2.5) node[left=1pt] {$\sum_{k=1}^i\delta_{4,A_4(k)}$};
  \draw[line width=0.7pt] (0.15, 3) -- (-0.15, 3) node[left=1pt] {$1$};
  \draw[line width=0.7pt] (0.15, 2) -- (-0.15, 2) node[left=1pt] {$0$};


  \foreach \x [evaluate=\x as \xx using {(\x+2)/2}] in {0,...,16} {
    \draw[grid style] (\xx, 0) -- ++(0, 9.5);
    \draw[line width=0.7pt] (\xx, 0.15) -- ++(0, -0.3)
      node[below=1pt] {\x};
  }

  \usetikzlibrary{patterns}
\begin{scope}
    \fill[red!30, opacity=0.3] (3.5, 0) rectangle (4, 9.5);
    \fill[red!30, opacity=0.3] (4.5, 0) rectangle (5, 9.5);

    \fill[blue!30, opacity=0.3] (6, 0) rectangle (6.5, 9.5);
    \fill[blue!30, opacity=0.3] (5, 0) rectangle (5.5, 9.5);
\end{scope}

\draw[grid style, red!90] (3.5, 0) -- ++(0, 9.5);
\draw[grid style, red!90] (4, 0) -- ++(0, 9.5);
\draw[grid style, red!90] (4.5, 0) -- ++(0, 9.5);
\draw[grid style, red!90] (5, 0) -- ++(0, 9.5);

\draw[<-, thick] (3.75,-1) arc (180:360:.5);
\draw[->, thick] (3.75,-1) arc (180:360:.5);

\draw[grid style, blue!90] (6.5, 0) -- ++(0, 9.5);
\draw[grid style, blue!90] (6, 0) -- ++(0, 9.5);
\draw[grid style, blue!90] (5, 0) -- ++(0, 9.5);
\draw[grid style, blue!90] (5.5, 0) -- ++(0, 9.5);

\draw[<-, thick] (5.25,-1) arc (180:360:.5);
\draw[->, thick] (5.25,-1) arc (180:360:.5);

  \def\Y{8.5} 
  \draw (\Xstart, \Y-\Amp) -- (\Xstart + \Xspan/2, \Y-\Amp) |- (\Xstart+ \Xspan/2, \Y+\Amp) |- (\Xend, \Y+\Amp)|- (\Xend, \Y-\Amp);

  \def\Y{6.5} 
  \draw (\Xstart, \Y-\Amp) -- (\Xstart + \Xspan/4, \Y-\Amp) |- (\Xstart + 3*\Xspan/4, \Y+\Amp) |- (\Xend, \Y-\Amp);

  \def\Y{4.5} 
  \draw (\Xstart, \Y-\Amp) -- (\Xstart + \Xspan/8, \Y-\Amp) |- (\Xstart + 5*\Xspan/16, \Y+\Amp)
  |- (\Xstart + 7*\Xspan/16, \Y-\Amp)
  |- (\Xstart + 9*\Xspan/16, \Y+\Amp)|- (\Xstart + 11*\Xspan/16, \Y-\Amp) |- (\Xstart + 7*\Xspan/8, \Y+\Amp)|- (\Xend, \Y-\Amp);

  \def\Y{2.5} 
  \draw (\Xstart, \Y-\Amp) -- (\Xstart + 1*\Xspan/16, \Y-\Amp)
      |- (\Xstart + 3*\Xspan/16, \Y+\Amp)
      |- (\Xstart + 6*\Xspan/16, \Y-\Amp)
      |- (\Xstart + 10*\Xspan/16, \Y+\Amp)
      |- (\Xstart + 13*\Xspan/16, \Y-\Amp)
      |- (\Xstart + 15*\Xspan/16, \Y+\Amp)
      |- (\Xend, \Y-\Amp);

\end{tikzpicture}
    \caption{The diagram we deformed for $n = 4$. Graphical expression through column permutation.}
    \label{fig:visual_proof_deform1}
\end{figure}

\begin{figure}

    \centering
    \begin{minipage}{0.5\textwidth}

    \begin{tikzpicture}[scale=0.75,
    every path/.style={color=black!90, line width=0.7pt},
    every node/.style={color=black!90}
  ]
  \def\Xstart{1}    
  \def\Xspan{8}     
  \def\Xend{\Xstart+\Xspan} 
  \def\Ymax{10}     
  \def\Amp{0.5}     

  \tikzstyle{grid style}=[densely dashed, color=black!60, line width=0.5pt]

  \draw[-stealth] (0,0) -- (\Xend+1.5, 0) node[below] {$i$};
  \draw[-stealth] (0, \Ymax) -- (0, 0);

  \node[left] at (0, 1) {$m$};

  \draw[line width=0.7pt] (0, 8.5) -- (-0, 8.5) node[anchor=east, inner sep=1pt] {};
  \draw[line width=0.7pt] (0.15, 9) -- (-0.15, 9) node[anchor=east, inner sep=1pt] {$1$};
  \draw[line width=0.7pt] (0.15, 8) -- (-0.15, 8) node[anchor=east, inner sep=1pt] {$0$};
  \draw[line width=0.7pt] (0, 6.5) -- (-0, 6.5) node[anchor=east, inner sep=1pt] {};
  \draw[line width=0.7pt] (0.15, 7) -- (-0.15, 7) node[anchor=east, inner sep=1pt] {$1$};
  \draw[line width=0.7pt] (0.15, 6) -- (-0.15, 6) node[anchor=east, inner sep=1pt] {$0$};
  \draw[line width=0.7pt] (0, 4.5) -- (-0, 4.5) node[anchor=east, inner sep=1pt] {};
  \draw[line width=0.7pt] (0.15, 5) -- (-0.15, 5) node[anchor=east, inner sep=1pt] {$1$};
  \draw[line width=0.7pt] (0.15, 4) -- (-0.15, 4) node[anchor=east, inner sep=1pt] {$0$};
  \draw[line width=0.7pt] (0, 2.5) -- (-0, 2.5) node[anchor=east, inner sep=1pt] {};
  \draw[line width=0.7pt] (0.15, 3) -- (-0.15, 3) node[anchor=east, inner sep=1pt] {$1$};
  \draw[line width=0.7pt] (0.15, 2) -- (-0.15, 2) node[anchor=east, inner sep=1pt] {$0$};

  \foreach \x [evaluate=\x as \xx using {(\x+2)/2}] in {0,...,16} {
    \draw[grid style] (\xx, 0) -- ++(0, 9.5);
    \draw[line width=0.7pt] (\xx, 0.15) -- ++(0, -0.3);
  }

  \def\Y{8.5} 
  \draw (\Xstart, \Y-\Amp) -- (\Xstart + \Xspan/2, \Y-\Amp) |- (\Xstart+ \Xspan/2, \Y+\Amp) |- (\Xend, \Y+\Amp)|- (\Xend, \Y-\Amp);

  \def\Y{6.5} 
  \draw (\Xstart, \Y-\Amp) -- (\Xstart + \Xspan/4, \Y-\Amp) |- (\Xstart + 3*\Xspan/4, \Y+\Amp) |- (\Xend, \Y-\Amp);
  \draw[green!90] (\Xstart+ \Xspan/4, \Y) -- (\Xstart + 3*\Xspan/4, \Y)
      -- (\Xstart + \Xspan/2, \Y+2) --(\Xstart+ \Xspan/4, \Y);
  
  \def\Y{4.5} 
  \draw (\Xstart, \Y-\Amp) -- (\Xstart + \Xspan/8, \Y-\Amp) |- (\Xstart + 5*\Xspan/16, \Y+\Amp)
  |- (\Xstart + 7*\Xspan/16, \Y-\Amp)
  |- (\Xstart + 9*\Xspan/16, \Y+\Amp)|- (\Xstart + 11*\Xspan/16, \Y-\Amp) |- (\Xstart + 7*\Xspan/8, \Y+\Amp)|- (\Xend, \Y-\Amp);

  \def\Y{2.5} 
  \draw (\Xstart, \Y-\Amp) -- (\Xstart + 1*\Xspan/16, \Y-\Amp)
      |- (\Xstart + 3*\Xspan/16, \Y+\Amp)
      |- (\Xstart + 6*\Xspan/16, \Y-\Amp)
      |- (\Xstart + 10*\Xspan/16, \Y+\Amp)
      |- (\Xstart + 13*\Xspan/16, \Y-\Amp)
      |- (\Xstart + 15*\Xspan/16, \Y+\Amp)
      |- (\Xend, \Y-\Amp);

  \draw[green!90] (\Xstart+\Xspan/16, \Y) -- (\Xstart + 2*\Xspan/16, \Y+2)
      -- (\Xstart + 3*\Xspan/16, \Y) --(\Xstart+\Xspan/16, \Y);
  \draw[green!90] (\Xstart+13*\Xspan/16, \Y) -- (\Xstart + 14*\Xspan/16, \Y+2)
      -- (\Xstart + 15*\Xspan/16, \Y) --(\Xstart+13*\Xspan/16, \Y);

  \draw[green!90] (\Xstart+\Xspan/16+\Xspan/4, 2+\Y) -- (\Xstart + 2*\Xspan/16+\Xspan/4, \Y)
      -- (\Xstart + 3*\Xspan/16+\Xspan/4, \Y+2) --(\Xstart+\Xspan/16+\Xspan/4, \Y+2);
  \draw[green!90] (\Xstart+13*\Xspan/16-+\Xspan/4, \Y+2) -- (\Xstart + 14*\Xspan/16-+\Xspan/4, \Y)
      -- (\Xstart + 15*\Xspan/16-+\Xspan/4, \Y+2) --(\Xstart+13*\Xspan/16-+\Xspan/4, \Y+2);
    \node[left] at (3.1+\Y, -\Y/3) {$(a)$};

\end{tikzpicture}

\end{minipage}%
\begin{minipage}{0.5\textwidth}
     \begin{tikzpicture}[scale=0.75,
    every path/.style={color=black!90, line width=0.7pt},
    every node/.style={color=black!90}
  ]
  \def\Xstart{1}
  \def\Xspan{8}
  \def\Xend{\Xstart+\Xspan}
  \def\Ymax{10}
  \def\Amp{0.5}

  \tikzstyle{grid style}=[densely dashed, color=black!60, line width=0.5pt]

  \draw[-stealth] (0,0) -- (\Xend+1.5, 0) node[below] {$i$};
  \draw[-stealth] (0, \Ymax) -- (0, 0);

  \node[left] at (0, 1) {$m$};

  \draw[line width=0.7pt] (0, 8.5) -- (-0, 8.5) node[anchor=east, inner sep=1pt] {};
  \draw[line width=0.7pt] (0.15, 9) -- (-0.15, 9) node[anchor=east, inner sep=1pt] {$1$};
  \draw[line width=0.7pt] (0.15, 8) -- (-0.15, 8) node[anchor=east, inner sep=1pt] {$0$};
  \draw[line width=0.7pt] (0, 6.5) -- (-0, 6.5) node[anchor=east, inner sep=1pt] {};
  \draw[line width=0.7pt] (0.15, 7) -- (-0.15, 7) node[anchor=east, inner sep=1pt] {$1$};
  \draw[line width=0.7pt] (0.15, 6) -- (-0.15, 6) node[anchor=east, inner sep=1pt] {$0$};
  \draw[line width=0.7pt] (0, 4.5) -- (-0, 4.5) node[anchor=east, inner sep=1pt] {};
  \draw[line width=0.7pt] (0.15, 5) -- (-0.15, 5) node[anchor=east, inner sep=1pt] {$1$};
  \draw[line width=0.7pt] (0.15, 4) -- (-0.15, 4) node[anchor=east, inner sep=1pt] {$0$};
  \draw[line width=0.7pt] (0, 2.5) -- (-0, 2.5) node[anchor=east, inner sep=1pt] {};
  \draw[line width=0.7pt] (0.15, 3) -- (-0.15, 3) node[anchor=east, inner sep=1pt] {$1$};
  \draw[line width=0.7pt] (0.15, 2) -- (-0.15, 2) node[anchor=east, inner sep=1pt] {$0$};

  \foreach \x [evaluate=\x as \xx using {(\x+2)/2}] in {0,...,16} {
    \draw[grid style] (\xx, 0) -- ++(0, 9.5);
    \draw[line width=0.7pt] (\xx, 0.15) -- ++(0, -0.3);
  }

  \def\Y{8.5}
  \draw (\Xstart, \Y-\Amp) -- (\Xstart + \Xspan/4, \Y-\Amp) |- (\Xstart + 3*\Xspan/4, \Y+\Amp) |- (\Xend, \Y-\Amp);
  \draw[green!90] (\Xstart+ \Xspan/4, \Y) -- (\Xstart + 3*\Xspan/4, \Y)
      -- (\Xstart + \Xspan/2, \Y-2) --(\Xstart+ \Xspan/4, \Y);
  
  \def\Y{6.5}
  \draw (\Xstart, \Y-\Amp) -- (\Xstart + \Xspan/2, \Y-\Amp) |- (\Xstart+ \Xspan/2, \Y+\Amp) |- (\Xend, \Y+\Amp)|- (\Xend, \Y-\Amp);

  \def\Y{2.5}
  \draw (\Xstart, \Y-\Amp) -- (\Xstart + \Xspan/8, \Y-\Amp) |- (\Xstart + 5*\Xspan/16, \Y+\Amp)
  |- (\Xstart + 7*\Xspan/16, \Y-\Amp)
  |- (\Xstart + 9*\Xspan/16, \Y+\Amp)|- (\Xstart + 11*\Xspan/16, \Y-\Amp) |- (\Xstart + 7*\Xspan/8, \Y+\Amp)|- (\Xend, \Y-\Amp);

  \def\Y{4.5}
  \draw (\Xstart, \Y-\Amp) -- (\Xstart + 1*\Xspan/16, \Y-\Amp)
      |- (\Xstart + 3*\Xspan/16, \Y+\Amp)
      |- (\Xstart + 6*\Xspan/16, \Y-\Amp)
      |- (\Xstart + 10*\Xspan/16, \Y+\Amp)
      |- (\Xstart + 13*\Xspan/16, \Y-\Amp)
      |- (\Xstart + 15*\Xspan/16, \Y+\Amp)
      |- (\Xend, \Y-\Amp);

  \draw[green!90] (\Xstart+\Xspan/16, \Y) -- (\Xstart + 2*\Xspan/16, \Y-2)
      -- (\Xstart + 3*\Xspan/16, \Y) --(\Xstart+\Xspan/16, \Y);
  \draw[green!90] (\Xstart+13*\Xspan/16, \Y) -- (\Xstart + 14*\Xspan/16, \Y-2)
      -- (\Xstart + 15*\Xspan/16, \Y) --(\Xstart+13*\Xspan/16, \Y);

  \draw[green!90] (\Xstart+\Xspan/16+\Xspan/4, \Y-2) -- (\Xstart + 2*\Xspan/16+\Xspan/4, \Y)
      -- (\Xstart + 3*\Xspan/16+\Xspan/4, \Y-2) --(\Xstart+\Xspan/16+\Xspan/4, \Y-2);
  \draw[green!90] (\Xstart+13*\Xspan/16-+\Xspan/4, \Y-2) -- (\Xstart + 14*\Xspan/16-+\Xspan/4, \Y)
      -- (\Xstart + 15*\Xspan/16-+\Xspan/4, \Y-2) --(\Xstart+13*\Xspan/16-+\Xspan/4, \Y-2);
    \node[left] at (1.0+\Y, -\Y/5) {$(b)$};

\end{tikzpicture}
\end{minipage}
    \caption{The same as a Fig. \ref{fig:visual_proof_deform1} diagram we deformed for $n = 4$. Graphical expression through geometry fractal intuition. As a generalization, for each triangle with a base from below, we first divide it into $4$ regions based on jumps at the edges, then put $2$ normal triangles in the extreme regions and $2$ inverted triangles in the central ones. For an inverted triangle, the opposite is true: inverted at the edges, and normal at the center.}
    \label{fig:visual_proof_deform1_triangle}
\end{figure}
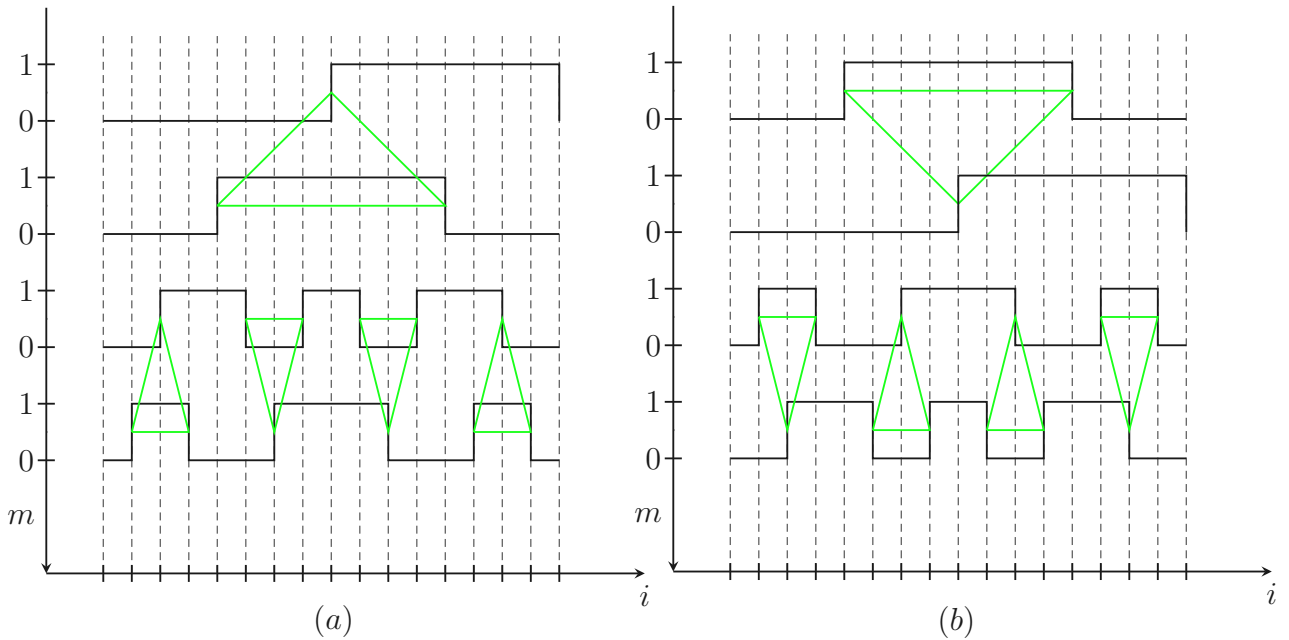

\begin{figure}
    \centering
\scalebox{0.7}{
\begin{tikzpicture}[x=1cm,y=1cm,
  treenode/.style={circle,draw=black!70,fill=white,inner sep=1pt,font=\small},
  arrowstyle/.style={-{Stealth[length=2mm,width=2mm]}, thick, draw=black!70}
]

\pgfmathsetmacro{\Xstart}{0}
\pgfmathsetmacro{\Xend}{24}
\pgfmathsetmacro{\Xspan}{\Xend - \Xstart}

\pgfmathsetmacro{\h}{1}
\pgfmathsetmacro{\Ytop}{0}
\pgfmathsetmacro{\rowsep}{1.5}

\pgfmathsetmacro{\Yone}{\Ytop}
\pgfmathsetmacro{\Ytwo}{\Ytop - \rowsep}
\pgfmathsetmacro{\Ythree}{\Ytop - 2*\rowsep}
\pgfmathsetmacro{\Ybase}{\Ythree - 2*\h}

\newcommand{\PlaceVertexAt}[4]{%
  \pgfmathsetmacro{\tempX}{#3}%
  \pgfmathsetmacro{\tempY}{#4}%
  \node[treenode] (#2) at (\tempX,\tempY) {#1};%
  \ifnum#1<7
    \pgfmathsetmacro{\Ybase}{\Ytop - 6*\rowsep}%
    \pgfmathsetmacro{\arrowlen}{\tempY - \Ybase-3*\rowsep}%
    \draw[arrowstyle] (#2.south) -- ++(0,-\arrowlen);
    \node[anchor=north] at ($(#2.south)+(0,-\arrowlen)$) {$ #1,$};
  \fi
}

\node[font=\large] at (0.125, -9.5+3*\rowsep) {$\{$};
\node[font=\large] at (\Xend-0.25, -9.5+3*\rowsep) {$\}$};

\coordinate (A1) at ({\Xstart+\Xspan/4},{\Yone-\h});
\coordinate (A2) at ({\Xstart+3*\Xspan/4},{\Yone-\h});
\coordinate (A3) at ({\Xstart+\Xspan/2},{\Yone});
\draw[green!90, thick] (A1) -- (A2) -- (A3) -- cycle;

\PlaceVertexAt{2}{N1A}{\Xstart+\Xspan/4}{\Yone-\h}
\PlaceVertexAt{2}{N1B}{\Xstart+3*\Xspan/4}{\Yone-\h}
\PlaceVertexAt{1}{N1C}{\Xstart+\Xspan/2}{\Yone}

\coordinate (B1)  at ({\Xstart+\Xspan/16},{\Ytwo-\h});
\coordinate (B2)  at ({\Xstart+3*\Xspan/16},{\Ytwo-\h});
\coordinate (B3)  at ({\Xstart+2*\Xspan/16},{\Ytwo});
\draw[green!90, thick] (B1) -- (B2) -- (B3) -- cycle;
\PlaceVertexAt{4}{N2A}{\Xstart+\Xspan/16}{\Ytwo-\h}
\PlaceVertexAt{4}{N2B}{\Xstart+3*\Xspan/16}{\Ytwo-\h}
\PlaceVertexAt{3}{N2C}{\Xstart+2*\Xspan/16}{\Ytwo}

\coordinate (B4)  at ({\Xstart+\Xspan/16+\Xspan/4},{\Ytwo});
\coordinate (B5)  at ({\Xstart+3*\Xspan/16+\Xspan/4},{\Ytwo});
\coordinate (B6)  at ({\Xstart+2*\Xspan/16+\Xspan/4},{\Ytwo-\h});
\draw[green!90, thick] (B4) -- (B5) -- (B6) -- cycle;
\PlaceVertexAt{3}{N2D}{\Xstart+\Xspan/16+\Xspan/4}{\Ytwo}
\PlaceVertexAt{3}{N2E}{\Xstart+3*\Xspan/16+\Xspan/4}{\Ytwo}
\PlaceVertexAt{4}{N2F}{\Xstart+2*\Xspan/16+\Xspan/4}{\Ytwo-\h}

\coordinate (B7)  at ({\Xstart+13*\Xspan/16},{\Ytwo-\h});
\coordinate (B8)  at ({\Xstart+15*\Xspan/16},{\Ytwo-\h});
\coordinate (B9)  at ({\Xstart+14*\Xspan/16},{\Ytwo});
\draw[green!90, thick] (B7) -- (B8) -- (B9) -- cycle;
\PlaceVertexAt{4}{N2G}{\Xstart+13*\Xspan/16}{\Ytwo-\h}
\PlaceVertexAt{4}{N2H}{\Xstart+15*\Xspan/16}{\Ytwo-\h}
\PlaceVertexAt{3}{N2I}{\Xstart+14*\Xspan/16}{\Ytwo}

\coordinate (B10) at ({\Xstart+13*\Xspan/16-\Xspan/4},{\Ytwo});
\coordinate (B11) at ({\Xstart+15*\Xspan/16-\Xspan/4},{\Ytwo});
\coordinate (B12) at ({\Xstart+14*\Xspan/16-\Xspan/4},{\Ytwo-\h});
\draw[green!90, thick] (B10) -- (B11) -- (B12) -- cycle;
\PlaceVertexAt{3}{N2J}{\Xstart+13*\Xspan/16-\Xspan/4}{\Ytwo}
\PlaceVertexAt{3}{N2K}{\Xstart+15*\Xspan/16-\Xspan/4}{\Ytwo}
\PlaceVertexAt{4}{N2L}{\Xstart+14*\Xspan/16-\Xspan/4}{\Ytwo-\h}

\pgfmathsetmacro{\unitL}{\Xspan/32}
\pgfmathsetmacro{\step}{2*\unitL}
\pgfmathsetmacro{\N}{16}
\pgfmathsetmacro{\blockW}{(2*\N - 1)*\unitL}
\pgfmathsetmacro{\xoffset}{(\Xspan - \blockW)/2}

\def\T{t}\def\B{b}

\foreach \dir [count=\idx from 0] in {b,t,t,b,t,b,b,t,t,b,b,t,b,t,t,b} {
  \pgfmathsetmacro{\xbase}{\Xstart + \xoffset + \idx * \step}
  \pgfmathsetmacro{\xmid}{\xbase + \unitL/2}
  \pgfmathsetmacro{\xend}{\xbase + \unitL}

  \ifx\dir\T
    \draw[green!90, thick] ({\xbase},{\Ythree}) -- ({\xend},{\Ythree}) -- ({\xmid},{\Ythree-\h}) -- cycle;
    \edef\nA{N3a\idx}\edef\nB{N3b\idx}\edef\nC{N3c\idx}
    \PlaceVertexAt{5}{\nA}{\xbase}{\Ythree}
    \PlaceVertexAt{5}{\nB}{\xend}{\Ythree}
    \PlaceVertexAt{6}{\nC}{\xmid}{\Ythree-\h}
  \else
    \draw[green!90, thick] ({\xbase},{\Ythree-\h}) -- ({\xmid},{\Ythree}) -- ({\xend},{\Ythree-\h}) -- cycle;
    \edef\nA{N3a\idx}\edef\nB{N3b\idx}\edef\nC{N3c\idx}
    \PlaceVertexAt{5}{\nA}{\xmid}{\Ythree}
    \PlaceVertexAt{6}{\nB}{\xbase}{\Ythree-\h}
    \PlaceVertexAt{6}{\nC}{\xend}{\Ythree-\h}
  \fi
}

\end{tikzpicture}
}

    \caption{An example of a tree that generates a sequence $A_n$ with a minimum number of $L(k)$}
    \label{fig:n6example}
\end{figure}
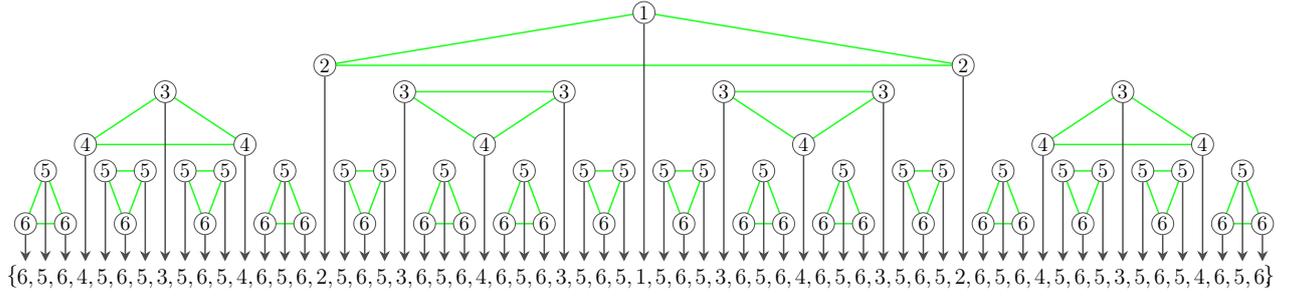

\subsection{Constant number of $L(k)$}

In this subsection, we will give an example of the sequence $A_n$, which is generated by rearranging the columns of the diagram from Fig. \ref{fig:visual_proof} while maintaining the number of 'gaps'.

First, let's consider the deformation of the diagram for $n=4$ and then generalize this to large $n$.
In Fig. \ref{fig:visual_proof_deform1} the columns that swap places are shown in the same colors. For convenience, the bottom is additionally duplicated using the arrows which columns we swap.

Since this permutation affects only the area of constancy of functions in the top two lines, it obviously acts identically on it. Before calculating the number of 'gaps', you can check that for this diagram in Fig. \ref{fig:visual_proof_deform1} the number of 'gaps' is the same as in the diagram \ref{fig:visual_proof}.

Now we have a diagram for $n=4$, let's generalize it to an arbitrary $n$. To do this, in Fig. \ref{fig:visual_proof_deform1_triangle}(a) it is convenient to draw triangles with vertices in the centers of 'gaps'. This gives us the first step in recompetitive circuit construction: we divide the entire diagram into 4 parts using vertical lines passing through the vertices of the upper triangle. In the bottom two lines, we draw triangles according to the following law: if the polygon falls into the extreme areas, it is drawn the same as the one above, and if it falls into the central areas, then we turn this triangle over. To continue the inverted triangle diagram, do the same, see Fig. \ref{fig:visual_proof_deform1_triangle}(b).

Following these rules, we can add two lines at each step of the recompetitive construction.

It is also worth noting how the sequence $A_n = a_n\circ\{1\}$ is generated. For such a structure, it is defined only for $n\mod2= 0$, since $2$ layers are added at once.

The first layer:
\begin{equation}
    A_2 = \{2,1,2,1\} = \{2,1,2\}\circ\{1\}
\end{equation}
On the second layer, we divide it into $4$ pieces and insert it into the two outermost $A_2+2$ and into the two central $5-A_2$:

\begin{equation}
    A_4 = \{4,3,4,2,3,4,3,1,3,4,3,2,4,3,4\}\circ\{1\}
\end{equation}

For inverted triangles, the formula changes to:
\begin{equation}
    a_{2inv} = \{1,2,1\}
\end{equation}
On the second layer, we divide it into $4$ pieces and insert it into the two outermost $a_2+2$ and into the two central $5-a_2$:

\begin{equation}
    a_{4inv} = \{3,4,3,1,4,3,4,2,4,3,4,1,3,4,3\}
\end{equation}

It is convenient to redefine $\overline{a_2} = 3-a_2 =\{2,1,2\}$. Then in general it looks like:
\begin{equation}
    \overline{a_{2^{n-1}}} = n+1-a_{2^{n-1}}\,\quad a_{n+2} = a_n\circ\{2\}\circ\overline{a_n}\circ\{1\}\circ\overline{a_n}\circ\{2\} \circ a_n\,\quad A_{n+2} = a_{n+2}\circ\{1\}
\end{equation}

A diagrammatic demonstration of this method for $n = 6$ is shown in Fig. \ref{fig:n6example}.

\subsection{Growing number of $L(k)$ }
If we now remove the condition for minimizing the number of $L(k)$ operators in the formula~\eqref{eq:U_tail_def}, this is equivalent to removing the condition for minimizing the number of 'gaps' in the diagram.
In this case, we use only the basic property of the diagram: its equivalence to the set $\{0,1\}^n$.
In this case, we obtain a non-degenerate scheme of a more general form. Now in the formula \eqref{eq:U_tail_def} for each operator $I_{2^{n-1}}\ot D_1(\beta_i)$ is not required to match only one operator $L(k)$. Let's make the necessary generalization of the formula \eqref{eq:U_tail_def} and the sequence $A_n$ for this case.

First of all, let's generalize definition of $A_n$:
\begin{definition}
    $\widetilde{A_n}$ - a sequence of $2^{n-1}$ sets. Formally:
    \begin{equation}
        \widetilde{A_n} = \{\forall i\in\{1,\dots,2^{n-1}\}\quad\widetilde{a}: \widetilde{a}\subset \{1,\dots,n-1\}\}
    \end{equation}
\end{definition}

The only difference between ${A_n}$ and $\widetilde{A_n}$ is that the first is a sequence of elements of the set $\{1,\dots,n-1\}$, and the second is a sequence of subsets of the same set $\{1,\dots,n-1\}$. Then instead of $L(A_n(i))$ in this generalization, we use $\prod_{k\in\widetilde{A_n}(i)}L(k)$.

With this definition, the formula for $U_{tail}$ is generalized as:

\begin{equation}
    U_{tail} = \prod_{i=1}^{2^{n-1}} \biggl( \Bigl(I_{2^{n-1}} \ot D_1(\beta_i)\Bigr) \cdot \prod_{k\in \widetilde{A_n}(i)} L(k) \biggr)
    \label{eq:U_tail_gen_def}
\end{equation}

In general, multiplication of all matrices from a certain set may have a different result depending on the order of multiplication. In order to understand that in our case this does not cause any problems, we will use:
\begin{prop}
    The generalization \eqref{eq:U_tail_gen_def} is correct

\end{prop}

\begin{proof}
     According to Prop. \ref{prop:L_commut} $[L(k),L(m)] = 0$ we can permute the matrices in the product as we like, then we can set the uniqueness of the matrix product by multiplying them in ascending order $k$.
\end{proof}

The previously written sequences take the form:
\begin{align}
    A_2 = \{1,1\} \Rightarrow \widetilde{A_2} &= \{\{1\}, \{1\}\};\\
    A_3 = \{2,1,2,1\} \Rightarrow \widetilde{A_3} &= \{\{2\},\{1\},\{2\}, \{1\}\};\\
    A_4 = \{3,2,3,1,3,2,3,1\} \Rightarrow \widetilde{A_4} &= \{\{3\},\{2\},\{3\},\{1\},\{3\},\{2\},\{3\},\{1\}\};\\
    \dots
\end{align}

Obviously, this generalization allows us to use previously unavailable tools - diagrams with multiple jumps at the same time.
For example, a diagram constructed simply as a fractal in which each next level is the two previous ones, in this case it will look as shown in Fig. \ref{fig:visual_proof_multi_spaces}. By construction, it meets the condition of non-degeneracy, because in the area of constancy of the value in the previous row, both values are equally present in the next row.

In this case, for such a diagram, the sequence is $\widetilde{A_n}$ will be:
\begin{align}
    \widetilde{A_2} &= \{\{1\},\{1\}\}\\
    \widetilde{A_3} &= \{\{2\},\{1,2\},\{2\},\{1,2\}\}\\
    \widetilde{A_4} &= \{\{3\},\{2,3\},\{3\},\{1,2,3\},\{3\},\{2,3\},\{3\},\{1,2,3\}\}\\
    \dots
\end{align}

Then you can set this sequence recursively using an intermediate set:

\begin{align}
    \overline{A_n} &= \{\varnothing \text{ if $i$ is odd, } \widetilde{A_n}(i/2) \text{ if $i$ is even} : i\in \{1,2,\dots,2^{n}\}\}    \\
    \widetilde{A_n} &= \{\{n-1\}\cup a: \forall a \in \overline{A_{n-1}}\}
\end{align}

For example, the first few sequences then look like this:

\begin{align}
    \widetilde{A_2} &= \{\{1\},\{1\}\}\\
    \overline{A_2} &= \{\varnothing, \{1\},\varnothing, \{1\}\}\\
    \widetilde{A_3} &= \{\{2\}\cup a: \forall a \in \overline{A_2}\} = \{\{2\},\{1,2\},\{2\},\{1,2\}\} \\
    \overline{A_3} &= \{\varnothing,\{2\},\varnothing,\{1,2\},\varnothing,\{2\},\varnothing,\{1,2\}\} \\
    \widetilde{A_4} &= \{\{3\}\cup a: \forall a \in \overline{A_3}\} = \{\{3\},\{2,3\},\{3\},\{1,2,3\},\{3\},\{2,3\},\{3\},\{1,2,3\}\} \\
    \dots
\end{align}

Thus, we have given an example of such a generalized sequence, given recursively, with a number of operators $L(k)$ greater than the reference one, but yielding a non-degenerate operator. This is an example of a diagram application in which we permuted the columns to get a diagram with a larger number of 'gaps'.
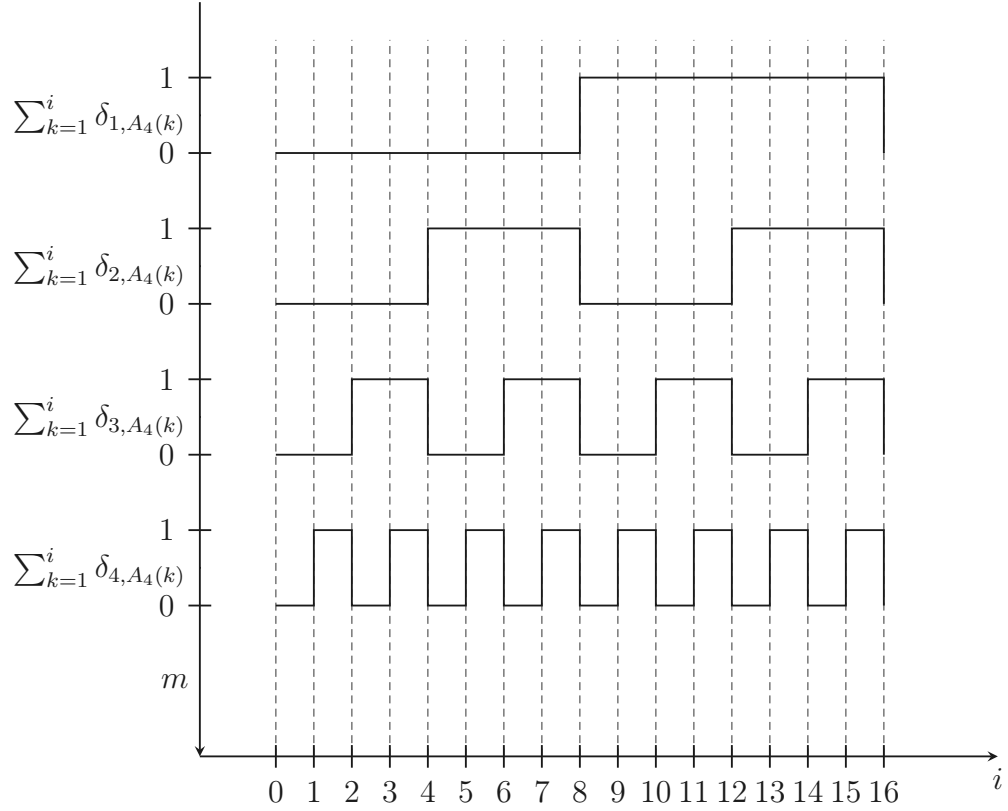
\begin{figure}

    \centering
    \begin{tikzpicture}[
    every path/.style={color=black!90, line width=0.7pt},
    every node/.style={color=black!90}
  ]
  \def\Xstart{1}    
  \def\Xspan{8}     
  \def\Xend{\Xstart+\Xspan} 
  \def\Ymax{10}     
  \def\Amp{0.5}     

  \tikzstyle{grid style}=[densely dashed, color=black!60, line width=0.5pt]

  \draw[-stealth] (0,0) -- (\Xend+1.5, 0) node[below] {$i$};
  \draw[-stealth] (0, \Ymax) -- (0, 0);

  \node[left] at (0, 1) {$m$};

\draw[line width=0.7pt] (0, 8.5) -- (-0, 8.5) node[left=1pt] {$\sum_{k=1}^i\delta_{1,A_4(k)}$};
  \draw[line width=0.7pt] (0.15, 9) -- (-0.15, 9) node[left=1pt] {$1$};
  \draw[line width=0.7pt] (0.15, 8) -- (-0.15, 8) node[left=1pt] {$0$};
  \draw[line width=0.7pt] (0, 6.5) -- (-0, 6.5) node[left=1pt] {$\sum_{k=1}^i\delta_{2,A_4(k)}$};
  \draw[line width=0.7pt] (0.15, 7) -- (-0.15, 7) node[left=1pt] {$1$};
  \draw[line width=0.7pt] (0.15, 6) -- (-0.15, 6) node[left=1pt] {$0$};
  \draw[line width=0.7pt] (0, 4.5) -- (-0, 4.5) node[left=1pt] {$\sum_{k=1}^i\delta_{3,A_4(k)}$};
  \draw[line width=0.7pt] (0.15, 5) -- (-0.15, 5) node[left=1pt] {$1$};
  \draw[line width=0.7pt] (0.15, 4) -- (-0.15, 4) node[left=1pt] {$0$};
  \draw[line width=0.7pt] (0, 2.5) -- (-0, 2.5) node[left=1pt] {$\sum_{k=1}^i\delta_{4,A_4(k)}$};
  \draw[line width=0.7pt] (0.15, 3) -- (-0.15, 3) node[left=1pt] {$1$};
  \draw[line width=0.7pt] (0.15, 2) -- (-0.15, 2) node[left=1pt] {$0$};


  \foreach \x [evaluate=\x as \xx using {(\x+2)/2}] in {0,...,16} {
    \draw[grid style] (\xx, 0) -- ++(0, 9.5);
    \draw[line width=0.7pt] (\xx, 0.15) -- ++(0, -0.3)
      node[below=1pt] {\x};
  }


  \def\Y{8.5} 
  \draw (\Xstart, \Y-\Amp) -- (\Xstart + \Xspan/2, \Y-\Amp) |- (\Xstart+ \Xspan/2, \Y+\Amp) |- (\Xend, \Y+\Amp)|- (\Xend, \Y-\Amp);

  \def\Y{6.5} 
  \draw (\Xstart, \Y-\Amp) -- (\Xstart + \Xspan/4, \Y-\Amp) |- (\Xstart+ \Xspan/4, \Y+\Amp) |- (\Xstart+\Xspan/2, \Y+\Amp)|-(\Xstart + \Xspan/2, \Y-\Amp) |- (\Xstart+ 3*\Xspan/4, \Y-\Amp) |- (\Xstart+\Xspan, \Y+\Amp)|- (\Xend, \Y-\Amp);
  \def\Y{4.5} 
  \draw (\Xstart, \Y-\Amp) -- (\Xstart + \Xspan/8, \Y-\Amp) |- (\Xstart+ \Xspan/8, \Y+\Amp) |- (\Xstart+\Xspan/4, \Y+\Amp)|-(\Xstart + \Xspan/4, \Y-\Amp) |- (\Xstart+ 3*\Xspan/8, \Y-\Amp) |- (\Xstart+\Xspan/2, \Y+\Amp) |- (\Xstart + \Xspan/8+\Xspan/2, \Y-\Amp) |- (\Xstart+\Xspan/2+ \Xspan/8, \Y+\Amp) |- (\Xstart+\Xspan/2+\Xspan/4, \Y+\Amp)|-(\Xstart +\Xspan/2+ \Xspan/4, \Y-\Amp) |- (\Xstart+\Xspan/2+ 3*\Xspan/8, \Y-\Amp) |- (\Xstart+\Xspan/2+\Xspan/2, \Y+\Amp)|- (\Xend, \Y-\Amp);

  \def\Y{2.5} 
  \draw (\Xstart, \Y-\Amp) -- (\Xstart + \Xspan/16, \Y-\Amp) |- (\Xstart+ \Xspan/16, \Y+\Amp) |- (\Xstart+\Xspan/8, \Y+\Amp)|-(\Xstart + \Xspan/8, \Y-\Amp) |- (\Xstart+ 3*\Xspan/16, \Y-\Amp) |- (\Xstart+\Xspan/4, \Y+\Amp) |- (\Xstart + \Xspan/16+\Xspan/4, \Y-\Amp) |- (\Xstart+\Xspan/4+ \Xspan/16, \Y+\Amp) |- (\Xstart+\Xspan/4+\Xspan/8, \Y+\Amp)|-(\Xstart +\Xspan/4+ \Xspan/8, \Y-\Amp) |- (\Xstart+\Xspan/4+ 3*\Xspan/16, \Y-\Amp) |- (\Xstart+\Xspan/4+\Xspan/4, \Y+\Amp)|-
   (\Xstart + \Xspan/16 + \Xspan/2, \Y-\Amp) |- (\Xstart + \Xspan/2 + \Xspan/16, \Y+\Amp) |- (\Xstart + \Xspan/2 + \Xspan/8, \Y+\Amp)|-(\Xstart + \Xspan/2 + \Xspan/8, \Y-\Amp) |- (\Xstart + \Xspan/2 + 3*\Xspan/16, \Y-\Amp) |- (\Xstart + \Xspan/2 + \Xspan/4, \Y+\Amp) |- (\Xstart + \Xspan/2 + \Xspan/16 + \Xspan/4, \Y-\Amp) |- (\Xstart  + \Xspan/2 + \Xspan/4 + \Xspan/16, \Y+\Amp) |- (\Xstart + \Xspan/2 + \Xspan/4 + \Xspan/8, \Y+\Amp)|-(\Xstart +\Xspan/4 + \Xspan/2 + \Xspan/8, \Y-\Amp) |- (\Xstart + \Xspan/2 + \Xspan/4+ 3*\Xspan/16, \Y-\Amp) |- (\Xstart + \Xspan/2 + \Xspan/4 + \Xspan/4, \Y+\Amp)|-
  (\Xend, \Y-\Amp);

\end{tikzpicture}
    \caption{A diagram generated by substituting a scaled-down copy of the previous row for $n = 4$ case. The generalization is obvious.}
    \label{fig:visual_proof_multi_spaces}
\end{figure}

%% file: parts/Chapter6.tex
\section{Conclusion}

We proved the decomposition of arbitrary diagonal unitary matrix $\un$  into a sequence of tensor and matrix products of diagonal matrices of a smaller size from  $\mathbf{SU}(4)$, $\mathbf{SU}(2)$ and $\mathbf{U}(1)$, proved its correctness, and constructed a linear non-degenerate mapping $\mathcal{L}$ between the parameters characterizing the original matrix and the matrices in the expansion set recursively. As we discussed in the Section \ref{sec:symmetries}, for a recursively given sequence, the number of $L(k)$ operators involved in each step is minimal and equal to the previous types of decomposition from \cite{crooks2024} and \cite{1629135}, and generalizes these approaches.

Initially, we have worked out a mathematical hypothesis that we obtained during a numerical experiment using machine learning methods \cite{FedinMorozov_unpublished}, the details of which were not considered in this paper.  We also examined all possible symmetries and obtained a convenient graphical representation that has a one-to-one correspondence between the step of the recurrent decomposition and the diagram. This is especially important in the context of quantum computing, as well as any fields that are sensitive to the order of operators, motivated by optimizing the number of operations, since our decomposition and its symmetries allow us to choose the most appropriate type of decomposition for a particular quantum algorithm in order to further optimize it. It is considered promising to generalize such diagrams in the case of decomposition of matrices from more exotic groups, as well as to generalize them to the entire scheme, since at the moment the diagram characterizes only the step of the recurrent decomposition, and not the entire decomposition at once.

\section{Acknowledgments}
The authors would like to thank Alexey Sleptsov, Dmitry Korzun, Dmitry Khudoteplov from LMTPh, Rail Khashaev from RQC and Ilya Sudakov from Yandex Research for helpful discussions. The work was supported by the state assignment of the Institute for Information Transmission Problems of RAS.

\newpage

%% file: parts/Appendix.tex
\section{Appendix}
\newcommand{\qnum}[1]{\left[#1\right]}
\newcommand{\brak}[1]{\left\langle #1 \right\rangle}
\newcommand{\A}{\mathcal{A}}
\newcommand{\SU}{\operatorname{SU}}

\subsection{$U(2^n)\Leftrightarrow U(N)$ dependence\label{sec:suNtosu2n}}

\begin{lemma}[Any matrix of $U(N)$ is a submatrix of the element $U(2^n)$]
Let $N < 2^n$. Then any unitary matrix $A_N\in U(N)$ can be embedded in the unitary matrix $A_{2^n}\in U(2^n)$ so that the set of all such nested matrices $A_{2^n}$ forms a subgroup in $U(2^n)$ isomorphic to $U(N)$.
\end{lemma}

\begin{proof}
Let's assume an arbitrary unitary matrix $A_N\in U(N)$, i.e. $A_N^\dagger A_N=I_N$. Let's build a block matrix:
\begin{equation}
    A_{2^n} = \begin{bmatrix}
        A_N & 0 \\
        0 & I_{2^n - N}
    \end{bmatrix} \in \mathbb{C}^{2^n \times 2^n}.
\end{equation}
Let's check the unitarity of $A_{2^n}$:
\begin{equation}
A_{2^n}^\dagger A_{2^n} =
\begin{bmatrix}
    A_N^\dagger & 0 \\
    0 & I
\end{bmatrix}
\begin{bmatrix}
    A_N & 0 \\
    0 & I
\end{bmatrix}
=
\begin{bmatrix}
    A_N^\dagger A_N & 0 \\
    0 & I
\end{bmatrix}
=
\begin{bmatrix}
    I_N & 0 \\
    0 & I_{2^n - N}
\end{bmatrix}
= I_{2^n},
\end{equation}

therefore, $A_{2^n} \in U(2^n)$.

Let's consider the composition of two such matrices:
\begin{equation}
    A_{2^n}B_{2^n} = 
    \begin{bmatrix}
        A_N & 0 \\
        0 & I
    \end{bmatrix}
    \begin{bmatrix}
        B_N & 0 \\
        0 & I
    \end{bmatrix}
    =
    \begin{bmatrix}
        A_NB_N & 0 \\
        0 & I
    \end{bmatrix},
\end{equation}
where $A_NB_N \in U(N)$.

Thus, the set of matrices of this type preserves the group structure: unity, reversibility, and closeness by multiplication are preserved. Therefore, the set of all such matrices $A_{2^n}$ is a subgroup in $U(2^n)$ that is isomorphic to $U(N)$.

All other group properties are inherited directly if these rules are followed.
\end{proof}

\subsection{Mutual decomposition Lemma \label{sec:j_to_rho}}

\begin{lemma}[]
\begin{equation}
    j(\rho) = \sum_{i=1}^{{n-1}}2^{n-1-i}\rho_{i} \Leftrightarrow \rho_i = \frac{j\mod 2^{n-i} - (j\mod2^{n+1-i}) }{2^{n-1-i}}
\end{equation}
\end{lemma}

\begin{proof}
\begin{enumerate}
    \item $\Leftarrow$ Substituting $\rho$ into the formula for $j(\rho)$, we obtain:
    \begin{multline}
        \sum_{i=1}^{{n-1}}2^{n-1-i} \frac{j\mod 2^{n-i} - (j\mod2^{n+1-i}) }{2^{n-1-i}} = \\ = 
        \sum_{i=1}^{{n-1}}j\mod 2^{n-i} - (j\mod2^{n+1-i}) =j\mod 2^{n-1} = j \text{ because } j<2^{n-1}
    \end{multline}

    \item $\Rightarrow$  Since:
    \begin{align}
        j\mod2^{n+1-i} = \sum_{j = i-1}^{n-1}2^{n-1-j}\rho_j; \quad
        j\mod2^{n-i} = \sum_{j = i}^{n-1}2^{n-1-j}\rho_j \Rightarrow \\
        \Rightarrow 2^{n-1-i}\rho_i = j\mod 2^{n-i} - (j\mod2^{n+1-i}) \Rightarrow \\ \Rightarrow \rho_i = \frac{j\mod 2^{n-i} - (j\mod2^{n+1-i}) }{2^{n-1-i}}
    \end{align}

\end{enumerate}
\end{proof}

%% file: main.bbl
\begin{thebibliography}{23}

\bibitem{aroyo2006bilbao}
Mois~I. Aroyo, Asen Kirov, Cesar Capillas, J.~M. Perez-Mato, and Hans
  Wondratschek.
\newblock Bilbao Crystallographic Server. II. Representations of
  crystallographic point groups and space groups.
\newblock {\em Acta Crystallographica Section A: Foundations of
  Crystallography}, 62(2):115--128, 2006.

\bibitem{BISHLER2023104729}
Liudmila Bishler, Andrei Mironov, and Andrey Morozov.
\newblock Invariants of knots and links at roots of unity.
\newblock {\em Journal of Geometry and Physics}, 185:104729, 2023.

\bibitem{crooks2024}
Gavin~E. Crooks.
\newblock Gates, states, and circuits: Quantum gates.
\newblock \url{https://threeplusone.com/gates}, March 2024.
\newblock Tech. Note 014 v0.11.0 beta.

\bibitem{FedinMorozov_unpublished}
M.~M. Fedin and A.~A. Morozov.
\newblock Machine learning approaches to building quantum circuits for groups
  of matrices.
\newblock Manuscript in preparation, 2025.

\bibitem{harder_ece250_pbt}
Douglas~Wilhelm Harder.
\newblock 4.05 perfect binary trees.
\newblock Lecture slides for ECE 250: Algorithms and Data Structures.
\newblock Hosted on the course page of C. Moreno.

\bibitem{huang2024optimizedsynthesiscircuitsdiagonal}
Xinchi Huang, Taichi Kosugi, Hirofumi Nishi, and Yu-ichiro Matsushita.
\newblock Optimized synthesis of circuits for diagonal unitary matrices with
  reflection symmetry, 2024.

\bibitem{kitaev1997}
A.~Yu. Kitaev.
\newblock Quantum computations: algorithms and error correction.
\newblock {\em Russian Mathematical Surveys}, 52(6):1191--1249, 1997.

\bibitem{kitaev2002}
Alexei~Yu. Kitaev, Alexander Shen, and Mikhail~N. Vyalyi.
\newblock {\em Classical and Quantum Computation}, volume~47 of {\em Graduate
  Studies in Mathematics}.
\newblock American Mathematical Society, Providence, Rhode Island, 2002.

\bibitem{KITAEV20032}
A.Yu. Kitaev.
\newblock Fault-tolerant quantum computation by anyons.
\newblock {\em Annals of Physics}, 303(1):2--30, 2003.

\bibitem{Knuth1997TAOCP1}
Donald~E. Knuth.
\newblock {\em The Art of Computer Programming, Volume 1: Fundamental
  Algorithms}.
\newblock Addison-Wesley Professional, Reading, MA, 3 edition, 1997.

\bibitem{KOLGANOV2023116072}
Nikita Kolganov, Sergey Mironov, and Andrey Morozov.
\newblock Large k topological quantum computer.
\newblock {\em Nuclear Physics B}, 987:116072, 2023.

\bibitem{Kolganov2020}
Nikita Kolganov and A.~Morozov.
\newblock Quantum {$\mathcal{R}$}-matrices as universal qubit gates.
\newblock {\em JETP Letters}, 111(9):519--524, May 2020.

\bibitem{expcnotsign}
Norbert~M. Linke, Dmitri Maslov, Martin Roetteler, Shantanu Debnath, Caroline
  Figgatt, Kevin~A. Landsman, Kenneth Wright, and Christopher Monroe.
\newblock Experimental comparison of two quantum computing architectures.
\newblock {\em Proceedings of the National Academy of Sciences},
  114(13):3305--3310, 2017.

\bibitem{Liu_2023}
Xingyi Liu and Keshab~K. Parhi.
\newblock Tensor decomposition for model reduction in neural networks: A
  review [feature].
\newblock {\em IEEE Circuits and Systems Magazine}, 23(2):8--28, 2023.

\bibitem{nch}
Michael~A. Nielsen and Isaac~L. Chuang.
\newblock {\em Quantum Computation and Quantum Information}.
\newblock Cambridge University Press, 10th anniversary edition, 2010.

\bibitem{doi:10.1137/24M1643451}
Luke Oeding and Ian Tan.
\newblock Tensor decompositions with applications to local unitary and
  stochastic local operations and classical communication equivalence of
  multipartite pure states.
\newblock {\em SIAM Journal on Applied Algebra and Geometry}, 9(1):33--57,
  2025.

\bibitem{doi:10.1137/090752286}
I.~V. Oseledets.
\newblock Tensor-train decomposition.
\newblock {\em SIAM Journal on Scientific Computing}, 33(5):2295--2317, 2011.

\bibitem{doi:10.1137/23M1592547}
Joshua Pickard, Can Chen, Cooper Stansbury, Amit Surana, Anthony~M. Bloch, and
  Indika Rajapakse.
\newblock Kronecker product of tensors and hypergraphs: Structure and dynamics.
\newblock {\em SIAM Journal on Matrix Analysis and Applications},
  45(3):1621--1642, 2024.

\bibitem{Schollw_ck_2011}
Ulrich Schollwöck.
\newblock The density-matrix renormalization group in the age of matrix product
  states.
\newblock {\em Annals of Physics}, 326(1):96--192, January 2011.

\bibitem{1629135}
V.V. Shende, S.S. Bullock, and I.L. Markov.
\newblock Synthesis of quantum-logic circuits.
\newblock {\em IEEE Transactions on Computer-Aided Design of Integrated
  Circuits and Systems}, 25(6):1000--1010, 2006.

\bibitem{turchetti2023decompositionlineartensortransformations}
Claudio Turchetti.
\newblock Decomposition of linear tensor transformations, 2023.

\bibitem{TYRTYSHNIKOV2004423}
Eugene Tyrtyshnikov.
\newblock Kronecker-product approximations for some function-related matrices.
\newblock {\em Linear Algebra and its Applications}, 379:423--437, 2004.
\newblock Special Issue on the Tenth ILAS Conference (Auburn, 2002).

\bibitem{Uotila_2024}
Valter Uotila.
\newblock Tensor decompositions and adiabatic quantum computing for discovering
  practical matrix multiplication algorithms.
\newblock In {\em 2024 IEEE International Conference on Quantum Computing and
  Engineering (QCE)}, pages 390--401. IEEE, September 2024.

\end{thebibliography}
